\documentclass[final,5p,times]{elsarticle}
\usepackage{graphicx,amssymb,lineno}
\usepackage{amsmath,amsfonts,amssymb,amsthm}
\usepackage{pifont}
\newtheorem{theorem}{\textbf{Theorem}}

\usepackage{enumerate}
\usepackage{algorithm}
\usepackage{algorithmic}
\usepackage[usenames]{color}
\usepackage{array}
\usepackage{multirow}
\usepackage{bm}
\usepackage{caption}
\usepackage{setspace}

\begin{document}
\begin{frontmatter}

\title{Time-Division Based Integrated Sensing, Communication, and Computing in Integrated Satellite-Terrestrial Networks}
\author[a]{Xiangming~Zhu}
\ead{zhuxm@zhejianglab.com}
\cortext[cor]{Corresponding author.}
\corref{cor}

\author[a]{Hua~Wang}

\author[a,b]{Zhaohui~Yang}

\author[c]{Quoc-Viet~Pham}

\address[a]{The Research Institute of Intelligent Networks, Zhejiang Lab, Hangzhou 311121, China}
\address[b]{The College of Information Science and Electronic Engineering, Zhejiang University, Hangzhou 310027, China}
\address[c]{The School of Computer Science and Statistics, Trinity College Dublin, University of Dublin, Dublin 2, D02 PN40, Ireland}

\vspace{-0mm}

\begin{abstract}
In this paper, we investigate time-division based framework for integrated sensing, communication, and computing in integrated satellite-terrestrial networks.
We consider a scenario, where Internet-of-Things devices on the ground operate with sensing and communication in a time-division manner, and can process the sensing results locally, at the edge, or in the cloud via the satellite communication link. Based on the proposed framework, we formulate a multi-dimensional optimization problem to maximize the utility performance of sensing, communication, and computing abilities.
After decomposing the original optimization problem into two subproblems, we first derive the closed-form solution of the optimal task partitioning strategy for terrestrial users and satellite users.
Then, we develop the joint subframe allocation and task partitioning strategy to optimize the overall performance, by means of which the Pareto optimal solutions can be obtained along the Pareto frontier.
Extensive simulations are provided to demonstrated the effectiveness of the proposed strategy, which is 10\% to 60\% superior compared with the benchmarks. Also, the trade-off between the multidimensional resource and multi-functional performance is analyzed from the perspective of network design.

\end{abstract}


\end{frontmatter}

\section{Introduction}\label{sec:1}

With the rapid development of information technology and devices, the traditional communication network is now evolving into an intelligent system \cite{Liu2022Path202,Chaccour2022less,Chen2023big}. The next-generation wireless communication network (6G) is expected to play an important role for large numbers of emerging applications, among which high-accuracy sensing capability is considered as one of the vital enablers \cite{Zhao2023semantic}, especially for the dramatic growth of the Internet-of-Things (IoT) \cite{Ma2020Sensing1222}.
Up to now, sensing and communication systems are generally developed in parallel with little integration. However, when it turns to the 6G network, the consistent development trend of sensing and communication systems, e.g. high spectrum and large-scale antenna arrays, provides a promising chance to integrate the two fundamental functionalities into one system \cite{Liu2022Integrated1728}. From this perspective, a lot of research efforts have been made for the integrated sensing and communication (ISAC) technology both in academia and industry. With the aim of integrated designing and full cooperation of the sensing and communication functionalities, the network resources can be better utilized and optimized, achieving lower hardware complexity and higher spectrum efficiency \cite{Martone2021View,Chaccour2022Seven967}.
Therefore, the ISAC technology is envisioned to bring a paradigm shift in the 6G network. Numerous new application scenarios have been foreseen, such as vehicular network, smart city, factory automation, and other environment-aware applications \cite{Liu2022Survey994}.

In the 6G era, the emerging intelligent services impose higher requirements for the sensing capability of the network \cite{Wei2022Learning9948}. More attentions have been focused on the processing of results instead of the original data itself \cite{Feng2021Joint34,Xu2023Edge9}. However, the limited computation resources of sensing terminals may hinder the potential use of advanced processing technologies, such as  deep learning algorithms and data mining, which are generally of high computational complexity \cite{Xiao2021UAV5933,Yang2020energy,Liu2023Toward158}.
Fortunately, by migrating part of the workload to the edge nodes, the multi-access edge computing (MEC) technology provides a promising solution for computation-intensive services, and then reduces the data processing delay for high performance sensing applications \cite{Jin2022survey,Chen2023IRS349,Yang2023rs}.
The integrated sensing, communication, and computing (ISCC) framework is thus motivated as an exciting research theme for the 6G network \cite{Qi2021Integrated332,Qi2022Integrating6212}.
In \cite{Zhao2022Radio8675}, the authors proposed a general design framework of ISCC, in which IoT devices implemented sensing and communication simultaneously with orthogonal frequency-division multiplexing (OFDM) waveforms, and the sensing results were offloaded to the base station (BS) for edge processing. The trade-off of sensing, communication, and computing performance was investigated based on the proposed subchannel allocation algorithm.
In \cite{Ding2022Joint2085}, a two-tier computing architecture was proposed for the ISCC system. The sensing results of user terminals can be processed locally, or offloaded to the BS. The network sensing performance and energy consumption were jointly optimized based on the proposed precoding design and the resource allocation schemes.
In \cite{Wang2023NOMA574}, both the edge computing and the cloud computing were integrated in the proposed ISCC architecture, for which the partial offloading mode was applied to enhance the computation efficiency.
In \cite{Xu2023UAV}, an unmanned aerial vehicle (UAV) based ISCC architecture was proposed. The UAV was considered to probe the sensing target on the ground with the sensing beam, while the sensing results were offloaded to the edge computing nodes with the offloading beam. The Pareto boundary of the sensing performance and the computing performance was analyzed based on the semidefinite programming method.
In \cite{Liu2022Energy1337,Liu2023Energy413}, the application of ISCC in vehicle-to-everything networks was investigated. By jointly optimizing the offloading strategy and the resource allocation, the environment information was obtained with high precision and low delay, providing reliable guarantee for the driving safety.

Currently, researches of the ISCC technology are mainly concentrated on the terrestrial network. Although great advancements have been experienced for the conventional terrestrial network in the past decades, it cannot satisfy the increasing global communication demand when it turns into the 6G era \cite{Wang20236G}. Due to the inherent limitations of terrestrial infrastructures, large numbers of populations, machines, sensors, and many other things still remain unconnected \cite{Yaacoub2020Key533}.
By providing extended coverage, the satellite network emerges as an optional approach to reinforce the connectivity in remote areas \cite{Hommsi20222Next18,Wang2023Satellite}. Also, the wide coverage of the satellite provides great potential for ISCC capabilities in wide areas \cite{Shi2023Intelligent147}. In \cite{You2022Beam2994}, the integrated sensing and communication was investigated in LEO satellite networks. By proposing a hybrid beam precoding algorithm, the sensing performance and the communication performance can be guaranteed simultaneously with relatively high efficiency.
On the other hand, it should be pointed out that conventional terrestrial networks have been investigated in depth for providing high speed services when covering densely populated areas. Thus, by combining the advantages of both satellite and terrestrial networks, the integrated satellite-terrestrial network architecture is proposed to facilitate ubiquitous and flexible services for the next-generation network \cite{Zhu2022Creating154,Ji2020Energy2265}. Plenty of works have focused on the integrated satellite-terrestrial network architecture from simple integration to deep cooperation \cite{Niephaus2016QoS2415,Xiao2022Mobility121}. In the White Paper of the 6G wireless network, it has also been proposed that the future wireless network must be able to seamlessly interface with terrestrial and satellite networks \cite{Latva2019Key}. Consequently, the application of ISCC to the integrated satellite-terrestrial network is an important issue and a comprehensive study is urgently needed.
In \cite{Zhao2023Integrated661}, the authors investigated the application of integrated sensing and communication in the integrated satellite-terrestrial network.
A dynamic resource allocation scheme was proposed for the sensing and communication phases to optimize the network throughput.

As discussed above, the research of the ISCC technology in the integrated satellite-terrestrial network is still in the early stage. Especially, the long propagation delay of satellite links leads to distinct network properties in the integrated satellite-terrestrial network, for which existing schemes for terrestrial networks cannot be directly applied. Also, the competition among terrestrial and satellite users adds to the complexity for resource management. Novel ISCC technologies are in great demand to fully exploit the cooperation of different network components in the integrated satellite-terrestrial network.
In this paper, we investigate the time-division based integrated sensing, communication, and computing in the integrated satellite-terrestrial network based on the ISCC architecture. The main contributions of this paper are summarized as follows.
\begin{itemize}
\item We propose a novel ISCC framework for the integrated satellite-terrestrial network, in which IoT devices on the ground operate with sensing and communication in a time-division manner, and can process the sensing results locally, at the edge, or in the cloud via the satellite  communication link. Based on the proposed framework, we formulate the joint subframe allocation and task partitioning problem to simultaneously optimize the sensing performance and minimize the processing delay of the sensing results.

\item We derive the optimal task partitioning strategy for both terrestrial users and satellite users. Based on the partial offloading model, the original optimization problem is decomposed into the terrestrial task partitioning subproblem and the satellite task partitioning subproblem. The closed-form solutions have been obtained for both subproblems with theoretical derivations.

\item We develop the joint subframe allocation and task partitioning strategy to optimize the overall performance of the network. Based on the coordination of different network components, the sensing, communication, and computing can be efficiently integrated in the proposed network architecture. Extensive simulations demonstrate the effectiveness of the proposed strategy, which is 10\% to 60\% superior compared with the benchmarks.

\end{itemize}

The rest of the paper is organized as follows. Section \ref{sec:2} describes the system model, and Section \ref{sec:3} discusses the problem formulation. In Section \ref{sec:4} and Section \ref{sec:5}, the optimal task partitioning strategies of terrestrial users and satellite users are derived, respectively. Then, the joint subframe allocation and task partitioning strategy is developed in Section \ref{sec:6}. Simulations and analyses are provided in Section \ref{sec:7}. Finally, conclusions are drawn in Section \ref{sec:8}.

\section{System Model} \label{sec:2}

\subsection{Network Model}\label{sec:2_1}
\begin{figure}[t]
\begin{center}
\includegraphics[width=0.95\linewidth]{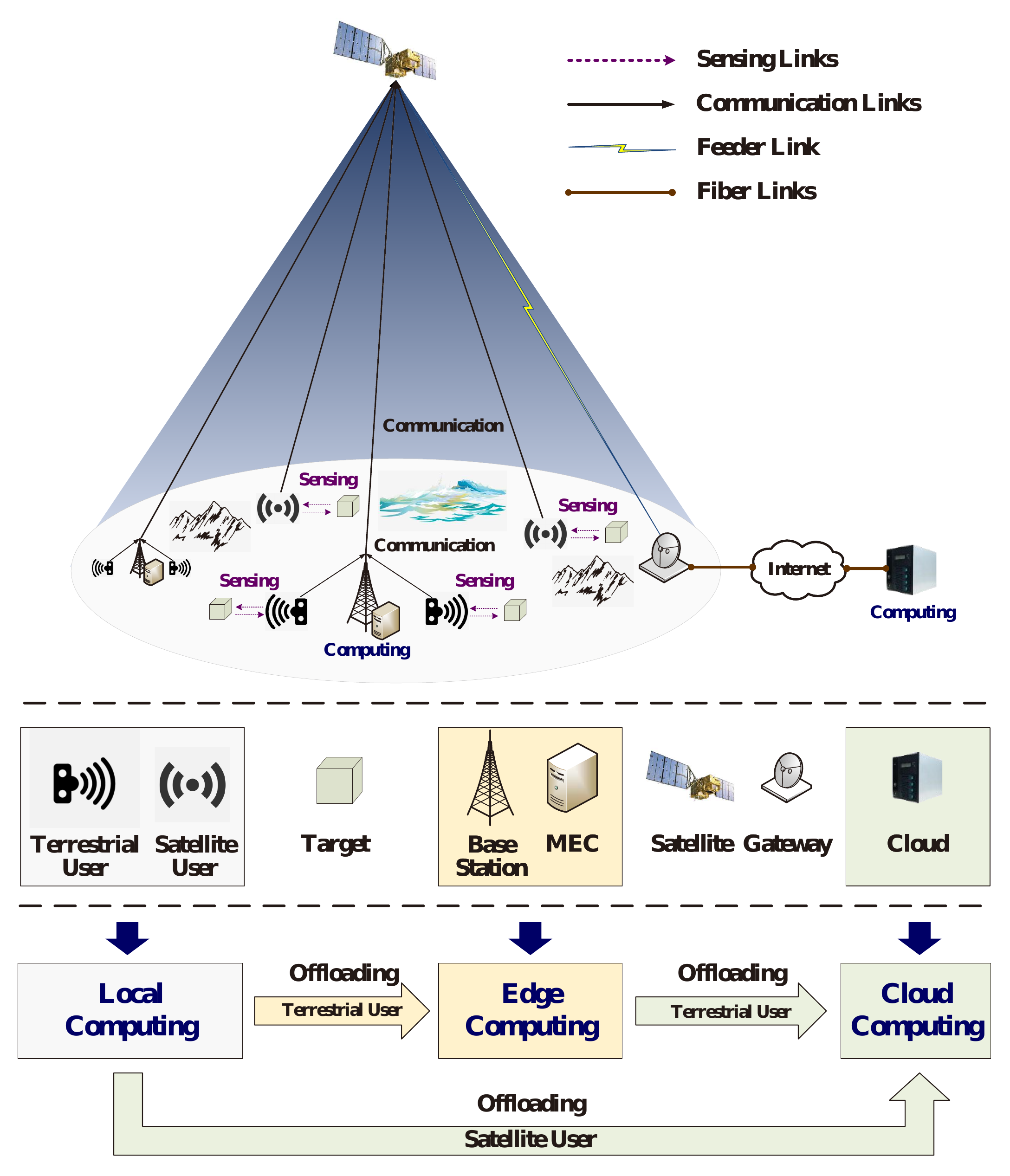}
\end{center}
\vspace{-5mm}
\captionsetup{font={footnotesize}}
\caption{The integrated satellite-terrestrial network with integrated sensing, communication, and computing.}
\label{fig:system}
\vspace{-5mm}
\end{figure}

The considered satellite-terrestrial network consists of one satellite, $N$ terrestrial BSs, and a set of IoT devices, as shown in Fig. \ref{fig:system}. The IoT devices are considered to be located in remote areas, and perform environment sensing for various environment-aware applications, such as weather prediction, marine monitoring, and pollution monitoring.
According to the communication link, IoT devices on the ground are divided into two types, i.e., users that use terrestrial networks (called TUEs hereinafter) and users that use satellite communications (called SUEs hereinafter). A set of $K_n^b$ TUEs that connect to BS $n$ is represented by ${\cal U}_n^b = \{ u_{n,1}^b,...,u_{n,k}^b,...,u_{n,K_n^b}^b\}$. The total number of TUEs is ${K^\textrm{B}}{\rm{ = }}\sum\limits_{n = 1}^N {K_n^b} $. A set of $K^\textrm{S}$ SUEs is represented by ${{\cal U}^\textrm{S}} = \{ u_1^s,...,u_k^s,...,u_{{K^\textrm{S}}}^s\}$. The total number of users in the network is $K = K^\textrm{B} + K^\textrm{S}$.

We adopt the satellite-terrestrial backhaul network
architecture, where the satellite provides the backhaul transmission for terrestrial BSs in remote areas without connection of optical fiber  \cite{Zhu2022Integrated437}. The satellite can be either low earth orbit (LEO) satellite, medium earth orbit (MEO) satellite, or geosynchronous earth orbit (GEO) satellite according to the actual network composition, which are of different transmission delay and service capability.
Equipped with an MEC server, each BS can provide edge computing services for its connected TUEs, while the cloud server can be accessed through the satellite for both TUEs and SUEs.
Both TUEs and SUEs implement integrated sensing and communication to improve the spectrum efficiency, and can process part of the sensing results with their own computation resources. Since the computing capability of IoT devices is generally limited, the computation tasks can be further offloaded to the BSs for edge computing, or offloaded to the cloud for cloud computing.

\vspace{-3mm}

\subsection{Sensing Model}\label{sec:2_2}
\subsubsection{TUE}

TUEs in the network implement sensing and communication in a time-division manner \cite{Cui2021Integrating158,Zhang2022Time2206}.
As shown in Fig. \ref{fig:frameBS}, the transmission frame of BS $n$ consists of sensing and communication subframes.

The set of TUEs ${\cal U}_n^b$ implement sensing during the sensing subframe $\theta _n^{{u_b},rad}{T_1} \in (0,T_1)$, in which $T_1$ is the length of a frame of the BS, and $\theta _n^{{u_b},rad}$ is the time fraction of TUEs for sensing. As discussed above, the BSs are isolated with each other in remote areas. Then, there will only be intra-BS interference for sensing of TUEs. Considering both the reflection and the refraction paths, the radar SINR of TUE $k$ can be obtained as follows \cite{Zhang2021Design1500}:
\begin{align}\label{eq:radar_SINR_BS_user}
\gamma _{n,k}^{{u_b},rad} = \frac{{g_{n,k,k}^{{u_b},rad}p_{n,k}^{{u_b}}}}{{\left(\sum\limits_{l = 1,l \ne k}^{K_n^b} {g_{n,l,k}^{{u_b},rad}p_{n,l}^{{u_b}}} \right) + {B_1}{N_0}}},
\end{align}
where $g_{n,k,k}^{{u_b},rad}$ is the propagation gain for the user $k$-target-user $k$ path, ${p_{n,k}^{{u_b}}}$ is the transmission power of TUE $k$, ${g_{n,l,k}^{{u_b},rad}}$ is the propagation gain for the user $l$-target-user $k$ path, ${p_{n,l}^{{u_b}}}$ is the transmission power of TUE $l$, $B_1$ is the bandwidth of the BS, and $N_0$ is the noise power spectral density.

\begin{figure}[t]
\begin{center}
\includegraphics[width=1\linewidth]{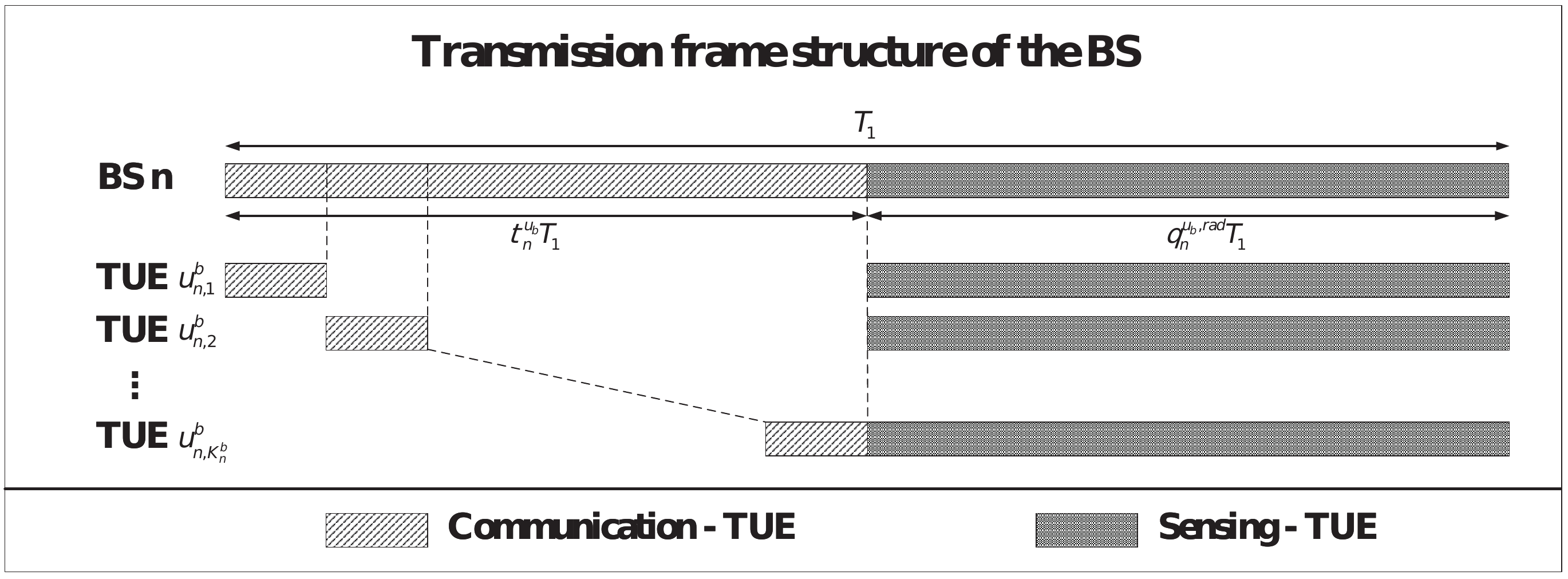}
\end{center}
\vspace{-5mm}
\captionsetup{font={footnotesize}}
\caption{The transmission frame structure of the BS.}
\label{fig:frameBS}
\vspace{-5mm}
\end{figure}

We use the radar mutual information (MI) to characterize the sensing performance of users, which is widely applied in ISAC and ISCC systems \cite{Zhang2022Enabling306,Ni2022Multi1276}. For TUE $k$, the sensing performance can be obtained by calculating the radar MI as follows:
\begin{align}\label{eq:MI_BS_user}
I_{n,k}^{{u_b}}(y_{n,k}^{{u_b},rad};g_{n,k,k}^{{u_b},rad}) = \theta _n^{{u_b},rad}{T_1}{B_1}{\log _2}(1 + \gamma _{n,k}^{{u_b},rad}),
\end{align}
where $y_{n,k}^{{u_b},rad}$ is the received sensing signal.

\subsubsection{SUE}

SUEs in the network implement sensing and communication in a time-division manner, as well as the uplink communication from the BSs to the satellite. As shown in Fig. \ref{fig:frameSatellite}, the transmission frame of the satellite consists of sensing and communication subframes.

Similarly to TUEs, the set of SUEs ${{\cal U}^\textrm{S}}$ implement sensing during the subframe ${\theta ^{{u_s},rad}}{T_2} \in (0,{T_2})$, in which $T_2$ is the length of a frame of the satellite, and ${\theta ^{{u_s},rad}}$ is the time fraction of SUEs for sensing. Also, we consider that SUEs in the network are isolated with each other. Then, there will be no interference for sensing of SUEs. The radar SINR of SUE $k$ can be obtained as follows:
\begin{align}\label{eq:radar_SINR_satellite_user}
\gamma _k^{{u_s},rad} = \frac{{g_{k,k}^{{u_s},rad}p_k^{{u_s}}}}{{{B_2}{N_0}}},
\end{align}
where $g_{k,k}^{{u_s},rad}$ is the propagation gain for the user $k$-target-user $k$ path, ${p_k^{{u_s}}}$ is the transmission power of SUE $k$, and $B_2$ is the bandwidth of the satellite.

\begin{figure}[t]
\begin{center}
\includegraphics[width=1\linewidth]{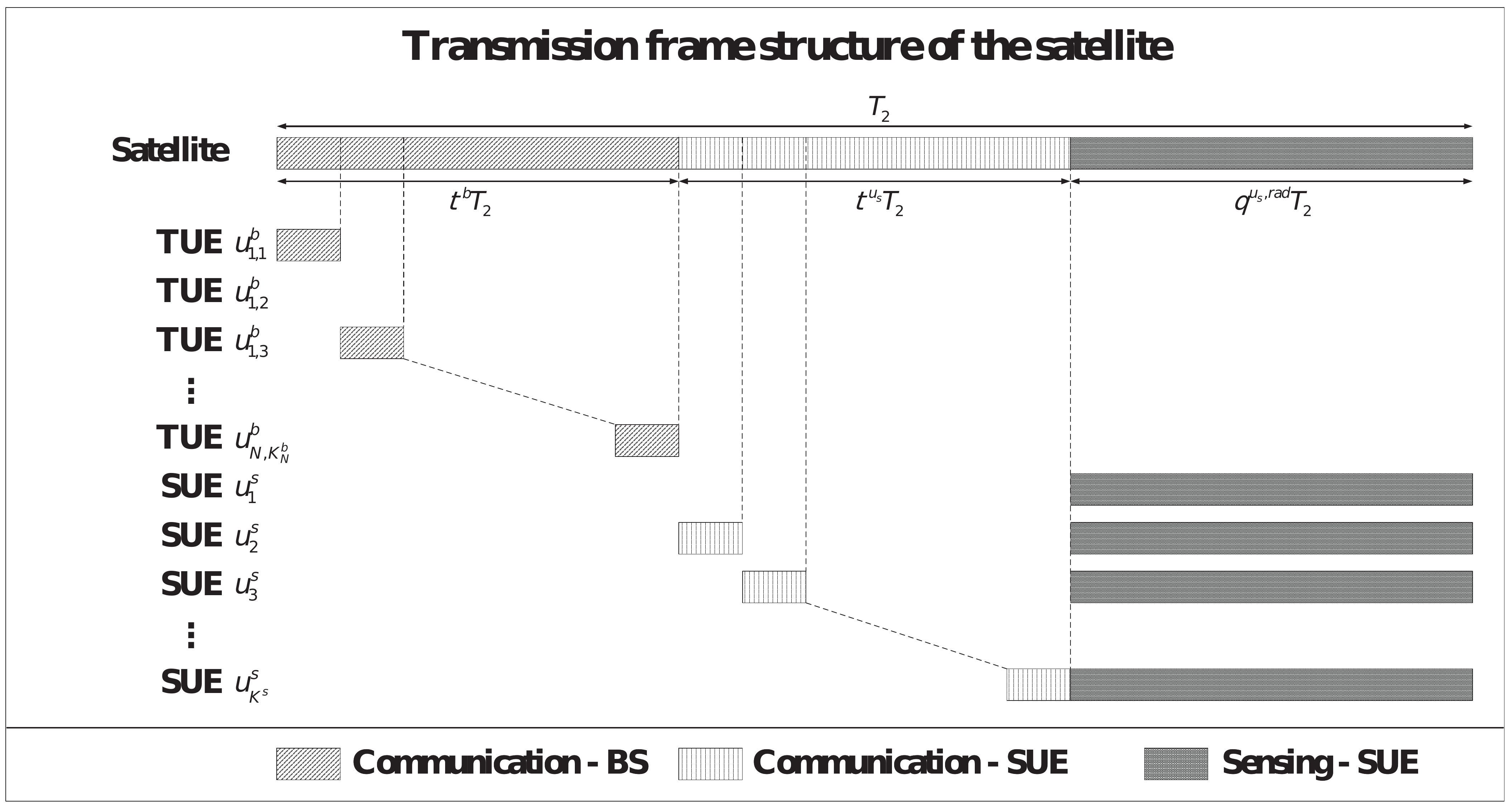}
\end{center}
\vspace{-5mm}
\captionsetup{font={footnotesize}}
\caption{The transmission frame structure of the satellite.}
\label{fig:frameSatellite}
\vspace{-5mm}
\end{figure}

Then, the sensing performance of SUE $k$ can be obtained by calculating the radar MI as follows:
\begin{align}\label{eq:MI_satellite_user}
I_k^{{u_s}}(y_k^{{u_s},rad};g_{k,k}^{{u_s},rad}) = {\theta ^{{u_s},rad}}{T_2}{B_2}{\log _2}(1 + \gamma _k^{{u_s},rad}),
\end{align}
where $y_k^{{u_s},rad}$ is the received sensing signal.

\subsection{Communication Model}\label{sec:2_3}

\subsubsection{TUE to BS}

As shown in Fig. \ref{fig:frameBS}, the duration of the communication subframe of TUEs is $\tau _n^{{u_b}}{T_1} \in (0,{T_1})$, in which $\tau _n^{{u_b}}$ is the time fraction of TUEs for communication, and we have $\tau _n^{{u_b}} + \theta _n^{{u_b},rad} = 1$. By applying the time-division multiple access (TDMA) method \cite{Lee2022Messaage34}, the communication subframe is equally divided among the set of TUEs ${\cal U}_n^b$ \cite{Fang2022Age1441,Mu2021Joint6648}. Then, the achievable transmission rate of TUE $k$ can be calculated by
\begin{align}\label{eq:rate_BS_user}
{r}_{n,k}^{{u_b}} = \frac{{\tau _n^{{u_b}}}}{{K_n^b}}{B_1}{\log _2}\left(1 + \frac{{G_{n,k}^{{u_b},t}G_n^{b,r}h_{n,k}^{{u_b}}p_{n,k}^{{u_b}}}}{{{B_1}{N_0}}}\right) = \frac{{\tau _n^{{u_b}}}}{{K_n^b}}R_{n,k}^{{u_b}},
\end{align}
where ${G_{n,k}^{{u_b},t}}$ is the transmitting antenna gain of TUE $k$, ${G_n^{b,r}}$ is the receiving antenna gain of BS $n$, and ${h_{n,k}^{{u_b}}}$ is the communication channel gain from TUE $k$ to BS $n$.

\subsubsection{BS to Satellite}
As shown in Fig. \ref{fig:frameSatellite}, the duration of the communication subframe of the BSs is $\tau^{b}{T_2} \in (0,{T_2})$, in which $\tau^{b}$ is the time fraction of BSs for communication. It is utilized to offload the computation tasks of TUEs to the cloud via the satellite. Let ${K^{{u_b},c}}$ be the number of TUEs whose computation tasks are offloaded to the cloud. By applying the TDMA method, the communication subframe of the BSs is equally divided for transmitting the computation tasks of the ${K^{{u_b},c}}$ TUEs. When offloading the computation task from BS $n$ to the satellite, the achievable transmission rate for the computation task of TUE $k$ can be calculated by
\begin{align}\label{eq:rate_BS}
r_{n,k}^b = \frac{{{\tau ^b}}}{{{K^{{u_b},c}}}}{B_2}{\log _2}\left(1 + \frac{{G_n^{b,t}G^{s,r}h_n^bp_n^b}}{{{B_2}{N_0}}}\right) = \frac{{{\tau ^b}}}{{{K^{{u_b},c}}}}R_n^b,
\end{align}
where ${G_n^{b,t}}$ is the transmitting antenna gain of BS $n$, ${G^{s,r}}$ is the receiving antenna gain of the satellite, ${h_n^b}$ is the communication channel gain from BS $n$ to the satellite, and ${p_n^b}$ is the transmission power of BS $n$.

\subsubsection{SUE to Satellite}
As shown in Fig. \ref{fig:frameSatellite}, the duration of the communication subframe of SUEs is $\tau^{u_s}{T_2} \in (0,{T_2})$, in which $\tau^{u_s}$ is the time fraction of SUEs for communication, and we have ${\tau ^b} + {\tau ^{{u_s}}} + {\theta ^{{u_s},rad}} = 1$. Let ${K^{{u_s},c}}$ be the number of SUEs whose computation tasks are offloaded to the cloud. By applying the TDMA method, the communication subframe of SUEs will be equally divided among the ${K^{{u_s},c}}$ SUEs. When offloading the computation task from SUE $k$ to the satellite, the achievable transmission rate can be calculated by
\begin{align}\label{eq:rate_satellite_user}
r_k^{{u_s}} = \frac{{{\tau ^{{u_s}}}}}{{{K^{{u_s},c}}}}{B_2}{\log _2}\left(1 + \frac{{G_k^{{u_s},t}{G^{s,r}}h_k^{{u_s}}p_k^{{u_s}}}}{{{B_2}{N_0}}}\right) = \frac{{{\tau ^{{u_s}}}}}{{{K^{{u_s},c}}}}R_k^{{u_s}},
\end{align}
where ${G_k^{{u_s},t}}$ is the transmitting antenna gain of SUE $k$, and ${h_k^{{u_s}}}$ is the communication channel gain from SUE $k$ to the satellite.

\subsubsection{Satellite to Cloud}
By means of high-gain directional antennas and large bandwidth at Ku/Ka band, high throughput can be achieved for the feeder link transmission between the satellite and the gateway \cite{Portillo2019Technical123}.
In this case, the transmission delay from the satellite to the gateway is negligibly small compared with the long propagation delay from the satellite to the cloud.
Thus we consider the propagation delay as the dominating factor when offloading the computation tasks from the satellite to the cloud, which will be further discussed in subsequent sections.

\subsection{Computation Model}

\subsubsection{Task Model}
We use a two-tuple model to characterize the computation task of users. For TUE $k$ of BS $n$, the computation task $Q_{n,k}^{{u_b}} = (D_{n,k}^{{u_b}},C_{n,k}^{{u_b}})$ consists of the data amount $D_{n,k}^{{u_b}}$ (in bits) and the computation workload $C_{n,k}^{{u_b}}$ (in CPU cycles/bit) \cite{Ren2019Collaborative5031}. Also, for SUE $k$, the computation task is denoted by $Q_k^{{u_s}} = (D_k^{{u_s}},C_k^{{u_s}})$.
Similar to existing works, the partial offloading mode is considered for the computing process  \cite{Wang2023NOMA574,Liu2020Post3170}. The computation task of users can be partitioned, and then executed at different locations in the network.

\subsubsection{Offloading Model - TUE}

As shown in Fig. \ref{fig:system}, the computing process of TUEs consists of the local computing, the edge computing at the BS, and the cloud computing. It should be noted that we do not consider satellite edge computing in this paper. Since the deployment and maintenance of the computing unit on the satellite are much more expensive than the terrestrial facility, it is inefficient to offload the computation task to the satellite compared with BS MEC servers. Thus the satellite is mainly utilized to provide connection with the cloud in the proposed network architecture.
For task $Q_{n,k}^{{u_b}}$, the computing process is given by:
\begin{itemize}
\item {\bf{Local Computing:}} The total amount of data executed at TUE $k$ is $\alpha _{n,k}^{{u_b}}D_{n,k}^{{u_b}}$, where $\alpha _{n,k}^{{u_b}} \in [0,1]$ is the task partitioning ratio for local computing. The local computing frequency of TUE $k$ is $f_{n,k}^{{u_b}}$ (in CPU cycles/s).

\item {\bf{Edge Computing:}} The total amount of data executed at BS $n$ is $\beta _{n,k}^{{u_b}}D_{n,k}^{{u_b}}$, where $\beta _{n,k}^{{u_b}} \in [0,1]$ is the task partitioning ratio for edge computing. We consider all TUEs will utilize edge computing to reduce the processing delay. The total computing frequency of BS $n$ is ${F^b}$, which will be equally allocated among the ${K_n^b}$ TUEs \cite{Zhao2022Radio8675,Ning2021Mobile463}.

\item {\bf{Cloud Computing:}} The total amount of data executed at the cloud is $\kappa _{n,k}^{{u_b}}D_{n,k}^{{u_b}}$, where $\kappa _{n,k}^{{u_b}} \in [0,1]$ is the task partitioning ratio for cloud computing, and we have $\alpha _{n,k}^{{u_b}} + \beta _{n,k}^{{u_b}} + \kappa _{n,k}^{{u_b}} = 1$.
    Then, the number of TUEs offloaded to the cloud can be expressed by
    \begin{align}\label{eq:cloud_number_BS_user}
    {K^{{u_b},c}} = \sum\limits_{n = 1}^N {\sum\limits_{k = 1}^{K_n^b} {\Gamma(\kappa _{n,k}^{{u_b}})} },
    \end{align}
    where $\Gamma(\kappa _{n,k}^{{u_b}})$ is the indicator function, given by
    \begin{align}\label{eq:indicator_function}
    \Gamma(\kappa _{n,k}^{{u_b}}) = \left\{ {\begin{array}{*{20}{c}} {1,\kappa _{n,k}^{{u_b}} > 0,}\\
    {0,\kappa _{n,k}^{{u_b}} = 0.}
    \end{array}} \right.
    \end{align}
    Generally, cloud servers are deployed with powerful computing resources. The computing delay at the cloud is considered to be negligibly small compared with the long propagation delay from the satellite to the cloud \cite{Guo2018Collaborative4514}.

\end{itemize}

\subsubsection{Offloading Model - SUE}

As shown in Fig. \ref{fig:system}, the computing process of SUEs consists of the local computing and the cloud computing. For task $Q_k^{{u_s}}$, the computing process is given by:
\begin{itemize}
\item {\bf{Local Computing:}} The total amount of data executed at SUE $k$ is $\alpha _k^{{u_s}}D_k^{{u_s}}$, where $\alpha _k^{{u_s}} \in [0,1]$ is the task partitioning ratio for local computing. The local computing frequency of SUE $k$ is $f_k^{{u_s}}$.

\item {\bf{Cloud Computing:}} The total amount of data executed at the cloud is $\kappa _k^{{u_s}}D_k^{{u_s}}$, where $\kappa _k^{{u_s}} \in [0,1]$ is the task partitioning ratio for cloud computing, and we have $\alpha _k^{{u_s}} + \kappa _k^{{u_s}} = 1$.
    Then, the number of SUEs offloaded to the cloud can be expressed by
    \begin{align}\label{eq:cloud_number_satellite_user}
    {K^{{u_s},c}} = \sum\limits_{k = 1}^{{K^s}} {\Gamma(\kappa _k^{{u_s}})}.
    \end{align}

\end{itemize}

\section{Problem Formulation}\label{sec:3}

\subsection{Delay Analysis - TUE}\label{sec:3_1}

\subsubsection{Local Computing}
For task $Q_{n,k}^{{u_b}}$, the total computation workload at the TUE is ${\alpha _{n,k}^{{u_b}}D_{n,k}^{{u_b}}C_{n,k}^{{u_b}}}$. Then, the local computing delay is given by
\begin{align}\label{eq:delay_local_computing_BS_user}
t_{n,k}^{{u_b},comp} = \frac{{\alpha _{n,k}^{{u_b}}D_{n,k}^{{u_b}}C_{n,k}^{{u_b}}}}{{f_{n,k}^{{u_b}}}}.
\end{align}

\subsubsection{Edge Computing}
For task $Q_{n,k}^{{u_b}}$, the total amount of data offloaded to BS $n$ is $(1 - \alpha _{n,k}^{{u_b}})D_{n,k}^{{u_b}}$. The uplink transmission delay is given by
\begin{align}\label{eq:transmission_delay_BS_user_BS}
t_{n,k}^{{u_b},tran} = \frac{{(1 - \alpha _{n,k}^{{u_b}})D_{n,k}^{{u_b}}}}{{r_{n,k}^{{u_b}}}} = \frac{{(1 - \alpha _{n,k}^{{u_b}})D_{n,k}^{{u_b}}K_n^b}}{{R_{n,k}^{{u_b}}\tau _n^{{u_b}}}}.
\end{align}
The computation workload at the BS is ${\beta _{n,k}^{{u_b}}D_{n,k}^{{u_b}}C_{n,k}^{{u_b}}}$. The edge computing delay is given by
\begin{align}\label{eq:delay_edge_computing_BS_user}
t_{n,k}^{b,comp} = \frac{{\beta _{n,k}^{{u_b}}D_{n,k}^{{u_b}}C_{n,k}^{{u_b}}K_n^b}}{{{F^b}}}.
\end{align}
With relatively small size, the download delay for the feedback of the computation results is generally considered to be negligible \cite{Wang2020Online803,Wang2022Deep}.

\subsubsection{Cloud Computing}
For task $Q_{n,k}^{{u_b}}$, the total amount of data offloaded to the satellite is $\kappa _{n,k}^{{u_b}}D_{n,k}^{{u_b}}$. The uplink transmission delay is given by
\begin{align}\label{eq:transmission_delay_BS_satellite}
t_{n,k}^{b,tran} = \frac{{\kappa _{n,k}^{{u_b}}D_{n,k}^{{u_b}}}}{{r_n^b}} = \frac{{\kappa _{n,k}^{{u_b}}D_{n,k}^{{u_b}}{K^{{u_b},c}}}}{{R_n^b{\tau ^b}}}.
\end{align}
Since the satellite is generally on the orbit of hundreds to tens of thousands of kilometers, the propagation delay of the offloading path should also be considered for cloud computing. The two-way propagation delay for computation offloading from BS $n$ to the cloud is given by
\begin{align}\label{eq:propagation_delay_BS_cloud}
t_n^{b,trip} = \frac{{2o_n^b}}{c} + \frac{{2{o^s}}}{c} + {t^{cloud}},
\end{align}
in which $o_n^b$ is the propagation path length for the BS $n$-satellite path, $o^s$ is the propagation path length for the satellite-gateway path, $c$ is the speed of light, and $t^{cloud}$ is the two-way propagation delay for the gateway-cloud path.

\subsubsection{Task Completion Delay}
Define ${t_{n,k}^{{u_b},u}}, {t_{n,k}^{{u_b},b}}, {t_{n,k}^{{u_b},c}}$ as the completion delays of the data executed at the TUE, at the BS, and at the cloud, which are given by
\begin{align}\label{eq:completion delay_data_location_BS_user}
&t_{n,k}^{{u_b},u} = t_{n,k}^{{u_b},comp},\\
&t_{n,k}^{{u_b},b} = t_{n,k}^{{u_b},tran} + t_{n,k}^{b,comp}, \nonumber \\
&t_{n,k}^{{u_b},c} = t_{n,k}^{{u_b},tran} + t_{n,k}^{b,tran} + t_n^{b,trip}. \nonumber
\end{align}
Then, the completion delay of task $Q_{n,k}^{{u_b}}$ can be calculated as follows \cite{Bi2022Energy1941}:
\begin{align}\label{eq:total_delay_BS_user}
&t_{n,k}^{{u_b},total}
= \left\{ {\begin{array}{*{20}{l}}
{\max \{ t_{n,k}^{{u_b},u},t_{n,k}^{{u_b},b}\} ,
\hspace{6.6mm} \hspace{5.7mm} \kappa _{n,k}^{{u_b}} = 0,}\\
{\max \{ t_{n,k}^{{u_b},u},t_{n,k}^{{u_b},b},t_{n,k}^{{u_b},c}\} ,
\hspace{5.7mm} \kappa _{n,k}^{{u_b}} > 0.}
\end{array}} \right.
\end{align}

\vspace{-5mm}

\subsection{Delay Analysis - SUE}
\vspace{-1mm}
\subsubsection{Local Computing}
For task $Q_k^{{u_s}}$, the computation workload at the SUE is $\alpha _k^{{u_s}}D_k^{{u_s}}$. Then, the local computing delay is given by
\begin{align}\label{eq:delay_local_computing_satellite_user}
t_k^{{u_s},comp} = \frac{{\alpha _k^{{u_s}}D_k^{{u_s}}C_k^{{u_s}}}}{{f_k^{{u_s}}}}.
\end{align}

\vspace{-3mm}
\subsubsection{Cloud Computing}
For task $Q_k^{{u_s}}$, the total amount of data offloaded to the satellite is ${\kappa _k^{{u_s}}D_k^{{u_s}}}$. The uplink transmission delay is given by
\begin{align}\label{eq:transmission_delay_satellite_user_satellite}
t_k^{{u_s},tran} = \frac{{\kappa _k^{{u_s}}D_k^{{u_s}}}}{{r_k^{{u_s}}}} = \frac{{\kappa _k^{{u_s}}D_k^{{u_s}}{K^{{u_s},c}}}}{{R_k^{{u_s}}{\tau ^{{u_s}}}}}.
\end{align}
Similarly, let ${o_k^{{u_s}}}$ be the propagation path length for the SUE $k$-satellite path. The two-way propagation delay from SUE $k$ to the cloud is given by
\begin{align}\label{eq:propagation_delay_satellite_user_cloud}
t_k^{{u_s},trip} = \frac{{2o_k^{{u_s}}}}{c} + \frac{{2{o^s}}}{c} + {t^{cloud}}.
\end{align}

\subsubsection{Task Completion Delay}
Define ${t_k^{{u_s},u}}$ and ${t_k^{{u_s},c}}$ as the completion delays of the data executed at the SUE and at the cloud, which are given by
\begin{align}\label{eq:completion delay_data_location_satellite_user}
&t_k^{{u_s},u} = t_k^{{u_s},comp},\\
&t_k^{{u_s},c} = t_k^{{u_s},tran} + t_k^{{u_s},trip} . \nonumber
\end{align}
Then, the completion delay of task $Q_k^{{u_s}}$ can be calculated as follows:
\begin{align}\label{eq:total_delay_satellite_user}
t_k^{{u_s},total} = \left\{ {\begin{array}{*{20}{l}}
{t_k^{{u_s},u},
\hspace{20.1mm} \kappa _k^{{u_s}} = 0,}\\
{\max \{ t_k^{{u_s},u},t_k^{{u_s},c}\} ,
\hspace{5.3mm} \kappa _k^{{u_s}} > 0.}
\end{array}} \right.
\end{align}

\subsection{Problem Formulation}

In the proposed integrated satellite-terrestrial network, the sensing performance is evaluated by the radar MI, which is determined by the duration of the sensing subframe. Also, the computing performance is evaluated by the task completion delay, which is determined by the duration of the communication subframe and the task partitioning ratio. Considering both the two performance functions, we introduce the Cobb-Douglas utility function to seek a trade-off between the radar MI and the task completion delay \cite{Ghosh2020Trade10914}, represented by
\begin{align}\label{eq:utlility_function}
{U^{total}} = \frac{{{{({I^{total}})}^{{\eta }}}}}{{{{({t^{total}})}^{{1- \eta}}}}} = \frac{{{{\left[\left(\sum\limits_{n = 1}^N {\sum\limits_{k = 1}^{K_n^b} {I_{n,k}^{{u_b}}} } \right) + \sum\limits_{k = 1}^{{K^s}} {I_k^{{u_s}}} \right]}^{{\eta }}}}}{{{{\left[\left(\sum\limits_{n = 1}^N {\sum\limits_{k = 1}^{K_n^b} {t_{n,k}^{{u_b},total}} } \right) + \sum\limits_{k = 1}^{{K^s}} {t_k^{{u_s},total}} \right]}^{{1- \eta}}}}},
\end{align}
where $I^{total}$ is the total radar MI, $t^{total}$ is the total task completion delay, and $\eta \in [0,1]$ is the weighting metric for characterizing the trade-off between the radar MI and the task completion delay. Larger value of $\eta$ means the radar MI is considered to be more important when optimizing the network performance, and vice versa. By adjusting the weighting metric $\eta$, the Pareto optimal solutions can be obtained along the Pareto frontier \cite{Censor1977Pareto41}.

Then, the optimization problem is formulated as
\begin{align}\label{eq:optimization_problem}
&{\cal P}:
\mathop {\max }\limits_{
\substack{
{{\bm{\tau }}^{{u_b}}},{{\bm{\theta }}^{{u_b},rad}},{\tau ^b},{\tau ^{{u_s}}},{\theta ^{{u_s},rad}},\\
{{\bm{\alpha }}^{{u_b}}},{{\bm{\beta }}^{{u_b}}},{{\bm{\kappa }}^{{u_b}}},{{\bm{\alpha }}^{{u_s}}},{{\bm{\kappa }}^{{u_s}}}}
} {U^{total}},\\
&\hspace{5.3mm} s.t.
\;\tau _n^{{u_b}} + \theta _n^{{u_b},rad} = 1,\tau _n^{{u_b}},\theta _n^{{u_b},rad} \ge 0,\forall n,
\tag{\ref{eq:optimization_problem}{a}} \label{eq:optimization_problem_a}\\
&\hspace{5.3mm}\;\;\;\;\;\; {\tau ^b} + {\tau ^{{u_s}}} + {\theta ^{{u_s},rad}} = 1,{\tau ^b},{\tau ^{{u_s}}},{\theta ^{{u_s},rad}} \ge 0,
\tag{\ref{eq:optimization_problem}{b}} \label{eq:optimization_problem_b}\\
&\hspace{5.3mm}\;\;\;\;\;\; \alpha _{n,k}^{{u_b}} + \beta _{n,k}^{{u_b}} + \kappa _{n,k}^{{u_b}} = 1,\alpha _{n,k}^{{u_b}},\beta _{n,k}^{{u_b}},\kappa _{n,k}^{{u_b}} \ge 0,\forall n,k,
\tag{\ref{eq:optimization_problem}{c}} \label{eq:optimization_problem_c}\\
&\hspace{5.3mm}\;\;\;\;\;\; \alpha _k^{{u_s}} + \kappa _k^{{u_s}} = 1,\alpha _k^{{u_s}},\kappa _k^{{u_s}} \ge 0,\forall k,
\tag{\ref{eq:optimization_problem}{d}} \label{eq:optimization_problem_d}
\end{align}
in which
${{\bm{\tau }}^{{u_b}}} = [\tau _1^{{u_b}},...,\tau _N^{{u_b}}],{{\bm{\theta }}^{{u_b},rad}} = [\theta _1^{{u_b}},...,\theta _N^{{u_b}}]$
are the durations of the subframe of the BSs,
${\tau ^b},{\tau ^{{u_s}}},{\theta ^{{u_s},rad}}$
are the durations of the subframe of the satellite,
${{\bm{\alpha }}^{{u_b}}} = [\alpha _{1,1}^{{u_b}},...,\alpha _{N,K_N^b}^{{u_b}}],{{\bm{\beta }}^{{u_b}}} = [\beta _{1,1}^{{u_b}},...,\beta _{N,K_N^b}^{{u_b}}],{{\bm{\kappa }}^{{u_b}}} = [\kappa _{1,1}^{{u_b}},...,\kappa _{N,K_N^b}^{{u_b}}]$
are the task partitioning ratios of TUEs,
and
${{\bm{\alpha }}^{{u_s}}} = [\alpha _1^{{u_s}},...,\alpha _{{K^s}}^{{u_s}}],{{\bm{\kappa }}^{{u_s}}} = [\kappa _1^{{u_s}},...,\kappa _{{K^s}}^{{u_s}}]$
are the task partitioning ratios of SUEs.
Constraint (\ref{eq:optimization_problem_a}) preserves the validity of subframe allocation of the BSs, constraint (\ref{eq:optimization_problem_b}) preserves the validity of subframe allocation of the satellite, constraint (\ref{eq:optimization_problem_c}) preserves the validity of task partitioning of TUEs, and constraint (\ref{eq:optimization_problem_d}) preserves the validity of task partitioning of SUEs.

\section{Task Partitioning Strategy of TUEs} \label{sec:4}

\subsection{Problem Decomposition}\label{sec:4_1}
Owing to the inter-coupled variables and the non-convex utility function, problem ${\cal P}$ in (\ref{eq:optimization_problem}) is of high complexity. However, since users in the network are considered to implement sensing and communication in a time-division manner, the coupling between sensing and communication is caused by the subframe allocation strategy. If we keep the subframe allocation strategy ${{{\bm{\tau }}^{{u_b}}},{{\bm{\theta }}^{{u_b},rad}},{\tau ^b},{\tau ^{{u_s}}},{\theta ^{{u_s},rad}}}$ fixed, the radar MI $I^{total}$ is also fixed based on (\ref{eq:MI_BS_user}) and (\ref{eq:MI_satellite_user}). Then, maximizing the utility $U^{total}$ in problem $\cal P$ is equal to minimize the task completion delay $t^{total}$. Also, when the subframe allocation strategy is fixed, the task completion delay of TUEs is only determined by the task partitioning strategy of TUEs, while the task partitioning strategy of TUEs will only influence the task completion delay of TUEs. This conclusion applies to SUEs similarly. Thus the task partitioning problem of TUEs and the task partitioning problem of SUEs are independent with each other in this case. The task partitioning strategy of TUEs can be obtained by solving the subproblem as follows:
\begin{align}\label{eq:subproblem_BS_user_all}
&{{\cal P}^{{u_b},delay}}:
\mathop {\min }\limits_{{{\bm{\alpha }}^{{u_b}}},{{\bm{\beta }}^{{u_b}}},{{\bm{\kappa }}^{{u_b}}}} {t^{{u_b},total}} = \sum\limits_{n = 1}^N {\sum\limits_{k = 1}^{K_n^b} {t_{n,k}^{{u_b},total}} }, \\
&\hspace{17.5mm} s.t.\;\alpha _{n,k}^{{u_b}} + \beta _{n,k}^{{u_b}} + \kappa _{n,k}^{{u_b}} = 1,\forall n,k, \nonumber \\
&\hspace{17.5mm}\;\;\;\;\;\; \alpha _{n,k}^{{u_b}},\beta _{n,k}^{{u_b}},\kappa _{n,k}^{{u_b}} \ge 0,\forall n,k. \nonumber
\end{align}

\vspace{-1mm}

In Section \ref{sec:3_1}, we analyzed the task completion delay of task $Q_{n,k}^{{u_b}}$ when being executed at different locations. Since the completion delay of task $Q_{n,k}^{{u_b}}$ in (\ref{eq:total_delay_BS_user}) is a discontinuous function, the task partitioning strategy of TUEs needs to be analyzed differently based on whether it is offloaded to the cloud.

\subsection{Local-Edge Task Partitioning Strategy of TUEs}\label{sec_4_2}

In this section, we first derive the local-edge task partitioning strategy of TUEs, if the computation task is not offloaded to the cloud. In this case, we have $t_{n,k}^{{u_b},total} = \max \{ t_{n,k}^{{u_b},u},t_{n,k}^{{u_b},b}\} $ and $\alpha _{n,k}^{{u_b}} + \beta _{n,k}^{{u_b}} = 1$ for task $Q_{n,k}^{{u_b}}$. Then, the partitioning strategy of task $Q_{n,k}^{{u_b}}$ can be calculated independently by solving the subproblem ${\cal P}_{n,k}^{{u_b},b}$, which is given by
\begin{align}\label{eq:subproblem_BS_user_each}
&{\cal P}_{n,k}^{{u_b},b}:
\mathop {\min }\limits_{\alpha _{n,k}^{{u_b}}} t_{n,k}^{{u_b},total}, \\
&\hspace{10.5mm} s.t.\;\alpha _{n,k}^{{u_b}} \in [0,1]. \nonumber
\end{align}

\begin{figure}[t]
\begin{center}
\includegraphics[width=1\linewidth]{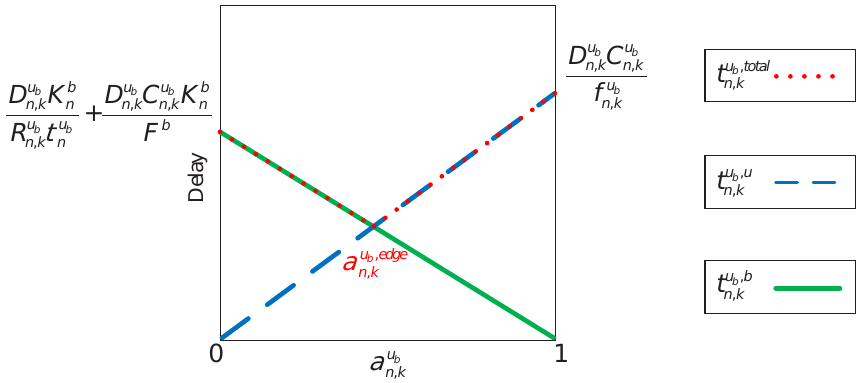}
\end{center}
\vspace{-5mm}
\captionsetup{font={footnotesize}}
\caption{The task completion delay of different task partitioning ratio.}
\label{fig:partitioning_ratio}
\vspace{-5mm}
\end{figure}

When $\alpha _{n,k}^{{u_b}}$ ranging between 0 and 1, the variation tendency of the task completion delay $t_{n,k}^{{u_b},total}$ is shown in Fig. \ref{fig:partitioning_ratio}. It can be seen that the task completion delay is minimized at the intersection. By solving the boundary condition $t_{n,k}^{{u_b},u} = t_{n,k}^{{u_b},b}$, the local-edge task partitioning strategy of task $Q_{n,k}^{{u_b}}$ is derived by
\begin{align}\label{eq:BS_user_task_partitioning_ratio_local_edge}
&\alpha _{n,k}^{{u_b},edge} = \dfrac{{\dfrac{{K_n^b}}{{R_{n,k}^{{u_b}}\tau _n^{{u_b}}}} + \dfrac{{C_{n,k}^{{u_b}}K_n^b}}{{{F^b}}}}}{{\dfrac{{K_n^b}}{{R_{n,k}^{{u_b}}\tau _n^{{u_b}}}} + \dfrac{{C_{n,k}^{{u_b}}K_n^b}}{{{F^b}}} + \dfrac{{C_{n,k}^{{u_b}}}}{{f_{n,k}^{{u_b}}}}}}, \\
&\beta _{n,k}^{{u_b},edge} = \dfrac{{\dfrac{{C_{n,k}^{{u_b}}}}{{f_{n,k}^{{u_b}}}}}}{{\dfrac{{K_n^b}}{{R_{n,k}^{{u_b}}\tau _n^{{u_b}}}} + \dfrac{{C_{n,k}^{{u_b}}K_n^b}}{{{F^b}}} + \dfrac{{C_{n,k}^{{u_b}}}}{{f_{n,k}^{{u_b}}}}}}. \nonumber
\end{align}
Then, the task completion delay of task $Q_{n,k}^{{u_b}}$ is obtained as
\begin{align}\label{eq:BS_user_task_completion_delay_local_edge}
&t_{n,k}^{{u_b},total,edge} = t_{n,k}^{{u_b},u,edge} = t_{n,k}^{{u_b},b,edge} \\
& \hspace{14mm} = t_{n,k}^{{u_b},comp,edge} = t_{n,k}^{{u_b},tran,edge} + t_{n,k}^{b,comp,edge}, \nonumber \\
&t_{n,k}^{{u_b},comp,edge} = \frac{{\alpha _{n,k}^{{u_b},edge}D_{n,k}^{{u_b}}C_{n,k}^{{u_b}}}}{{f_{n,k}^{{u_b}}}}, ~~
t_{n,k}^{{u_b},tran,edge} = \frac{{\beta _{n,k}^{{u_b},edge}D_{n,k}^{{u_b}}K_n^b}}{{R_{n,k}^{{u_b}}\tau _n^{{u_b}}}}, \nonumber \\
&t_{n,k}^{b,comp,edge} = \frac{{\beta _{n,k}^{{u_b},edge}D_{n,k}^{{u_b}}C_{n,k}^{{u_b}}K_n^b}}{{{F^b}}}. \nonumber
\end{align}

For any task $Q_{n,k}^{{u_b}}$, if it is not offloaded to the cloud, the local-edge task partitioning strategy is derived by (\ref{eq:BS_user_task_partitioning_ratio_local_edge}). However, to solve the original subproblem ${{\cal P}^{{u_b},delay}}$ in (\ref{eq:subproblem_BS_user_all}), the task partitioning strategy needs to be derived considering the cooperation of the local computing, the edge computing, and the cloud computing.
If we further offload task $Q_{n,k}^{{u_b}}$ to the cloud via the satellite, although the edge computing delay can be reduced, it will lead to an extra propagation delay due to the long transmission distance. Obviously, offloading task $Q_{n,k}^{{u_b}}$ to the cloud should not increase the delay. If the cost brought by cloud computing is larger than the benefit obtained, it is better to fully process the task at the user device and the BS. Thus we discuss the task splitting strategy by comparing the edge computing delay $t_{n,k}^{b,comp,edge}$ with the propagation delay $t_n^{b,trip}$. The following two cases will be discussed in Section \ref{sec:4_3} and Section \ref{sec:4_4} separately.

\subsection{Optimal Strategy Discussion: When $t_{n,k}^{b,comp,edge} \le t_n^{b,trip}$}\label{sec:4_3}
Assume that task $Q_{n,k}^{{u_b}}$ is offloaded to the cloud, and the total amount of data executed at the cloud is ${\kappa _{n,k}^{{u_b}}D_{n,k}^{{u_b}}}$. The task completion delay of task $Q_{n,k}^{{u_b}}$ is calculated by $t_{n,k}^{{u_b},total} = \max \{ t_{n,k}^{{u_b},u},t_{n,k}^{{u_b},b},t_{n,k}^{{u_b},c}\}$. We will then prove $t_{n,k}^{{u_b},total} > t_{n,k}^{{u_b},total,edge}$ by analyzing the value of $\alpha _{n,k}^{{u_b}}$.

\begin{itemize}
\item Case 1: $\alpha _{n,k}^{{u_b}} > \alpha _{n,k}^{{u_b},edge}$

    In the case of $\alpha _{n,k}^{{u_b}} > \alpha _{n,k}^{{u_b},edge}$, it can be obtained that
    \begin{align}\label{eq:BS_user_case1_alpha_large_local}
    t_{n,k}^{{u_b},comp} = \frac{{\alpha _{n,k}^{{u_b}}D_{n,k}^{{u_b}}C_{n,k}^{{u_b}}}}{{f_{n,k}^{{u_b}}}} > \frac{{\alpha _{n,k}^{{u_b},edge}D_{n,k}^{{u_b}}C_{n,k}^{{u_b}}}}{{f_{n,k}^{{u_b}}}} = t_{n,k}^{{u_b},comp,edge}.
    \end{align}
    Then, we can conclude that
    \begin{align}\label{eq:BS_user_case1_alpha_large_total}
    t_{n,k}^{{u_b},total} \ge t_{n,k}^{{u_b},u} = t_{n,k}^{{u_b},comp} > t_{n,k}^{{u_b},comp,edge} = t_{n,k}^{{u_b},total,edge}.
    \end{align}
\item Case 2: $\alpha _{n,k}^{{u_b}} \le \alpha _{n,k}^{{u_b},edge}$

    In the case of $\alpha _{n,k}^{{u_b}} \le \alpha _{n,k}^{{u_b},edge}$, it can be obtained that
    \begin{align}\label{eq:BS_user_case1_alpha_small_usertoBS}
    t_{n,k}^{{u_b},tran}
    &= \frac{{(1 - \alpha _{n,k}^{{u_b}})D_{n,k}^{{u_b}}K_n^b}}{{R_{n,k}^{{u_b}}\tau _n^{{u_b}}}} \\
    &\ge \frac{{(1 - \alpha _{n,k}^{{u_b},edge})D_{n,k}^{{u_b}}K_n^b}}{{R_{n,k}^{{u_b}}\tau _n^{{u_b}}}} = t_{n,k}^{{u_b},tran,edge}. \nonumber
    \end{align}
    Since $t_{n,k}^{b,comp,edge} \le t_n^{b,trip}$, the following inequality holds for task $Q_{n,k}^{{u_b}}$ when $\kappa _{n,k}^{{u_b}} > 0$.
    \begin{align}\label{eq:BS_user_case1_alpha_small_BStocloud}
    t_{n,k}^{b,tran} + t_n^{b,trip} = \frac{{\kappa _{n,k}^{{u_b}}D_{n,k}^{{u_b}}{K^{{u_b},c}}}}{{R_n^b{\tau ^b}}} + t_n^{b,trip} > t_{n,k}^{b,comp,edge}.
    \end{align}
    Then, we can conclude that
    \begin{align}\label{eq:BS_user_case1_alpha_small_total}
    t_{n,k}^{{u_b},total}  \ge t_{n,k}^{{u_b},c}
    & = t_{n,k}^{{u_b},tran} + t_{n,k}^{b,tran} + t_n^{b,trip} \\
    & > t_{n,k}^{{u_b},tran,edge} + t_{n,k}^{b,comp,edge} = t_{n,k}^{{u_b},total,edge}. \nonumber
    \end{align}

\end{itemize}

When $t_{n,k}^{b,comp,edge} \le t_n^{b,trip}$, we prove that $t_{n,k}^{{u_b},total} > t_{n,k}^{{u_b},total,edge}$ holds for any $\kappa _{n,k}^{{u_b}} > 0$. Thus task $Q_{n,k}^{{u_b}}$ should not be offloaded to the cloud. The optimal task partitioning strategy is given by $\alpha _{n,k}^{{u_b}*} = \alpha _{n,k}^{{u_b},edge},\beta _{n,k}^{{u_b}*} = \beta _{n,k}^{{u_b},edge},\kappa _{n,k}^{{u_b}*} = 0$.

\subsection{Optimal Strategy Discussion: When $t_{n,k}^{b,comp,edge} > t_n^{b,trip}$}\label{sec:4_4}
When $t_{n,k}^{b,comp,edge} > t_n^{b,trip}$, there will always exist a task partitioning strategy with $\kappa _{n,k}^{{u_b}} > 0$, for which the task completion delay of task $Q_{n,k}^{{u_b}}$ satisfies $t_{n,k}^{{u_b},total} < t_{n,k}^{{u_b},total,edge}$. This conclusion can be directly derived based on the expressions of $t_{n,k}^{{u_b},total}$ and $t_{n,k}^{{u_b},total,edge}$ in (\ref{eq:total_delay_BS_user}) and (\ref{eq:BS_user_task_completion_delay_local_edge}).
In this case, it is possible to achieve lower task completion delay by offloading the task to the cloud.
In Section \ref{sec_4_2}, the local-edge task partitioning strategy of task $Q_{n,k}^{{u_b}}$ is derived by solving the boundary condition $t_{n,k}^{{u_b},u} = t_{n,k}^{{u_b},b}$. By extending this conclusion to the local-edge-cloud computing case, Theorem \ref{theorem:task_partitioning_optimality} is derived and proved as follows.

\begin{theorem}\label{theorem:task_partitioning_optimality}
For task $Q_{n,k}^{{u_b}}$, if $t_{n,k}^{b,comp,edge} > t_n^{b,trip}$, the task completion delay $t_{n,k}^{total}$ is minimized when the boundary condition $t_{n,k}^{{u_b},u} = t_{n,k}^{{u_b},b} = t_{n,k}^{{u_b},c}$ is satisfied.
\end{theorem}
\begin{proof}[\bfseries Proof]
See Appendix A.
\end{proof}

Based on Theorem \ref{theorem:task_partitioning_optimality}, the local-edge-cloud partitioning strategy of task $Q_{n,k}^{{u_b}}$ can be obtained by solving the boundary condition $t_{n,k}^{{u_b},u} = t_{n,k}^{{u_b},b} = t_{n,k}^{{u_b},c}$, which is given by
\begin{align}\label{eq:BS_user_task_partitioning_ratio_edge_cloud}
&\alpha _{n,k}^{{u_b},cloud} = \\ &\dfrac{{\dfrac{{D_{n,k}^{{u_b}}K_n^b}}{{R_{n,k}^{{u_b}}\tau _n^{{u_b}}}} + \dfrac{{D_{n,k}^{{u_b}}C_{n,k}^{{u_b}}K_n^b{K^{{u_b},c}}}}{{C_{n,k}^{{u_b}}R_n^b{\tau ^b}K_n^b + {F^b}{K^{{u_b},c}}}} + \dfrac{{t_n^{b,trip}C_{n,k}^{{u_b}}R_n^b{\tau ^b}K_n^b}}{{C_{n,k}^{{u_b}}R_n^b{\tau ^b}K_n^b + {F^b}{K^{{u_b},c}}}}}}{{\dfrac{{D_{n,k}^{{u_b}}K_n^b}}{{R_{n,k}^{{u_b}}\tau _n^{{u_b}}}} + \dfrac{{D_{n,k}^{{u_b}}C_{n,k}^{{u_b}}K_n^b{K^{{u_b},c}}}}{{C_{n,k}^{{u_b}}R_n^b{\tau ^b}K_n^b + {F^b}{K^{{u_b},c}}}} + \dfrac{{D_{n,k}^{{u_b}}C_{n,k}^{{u_b}}}}{{f_{n,k}^{{u_b}}}}}}, \nonumber \\
&\beta _{n,k}^{{u_b},cloud} = \dfrac{{\dfrac{{(1 - \alpha _{n,k}^{{u_b},cloud})D_{n,k}^{{u_b}}{K^{{u_b},c}}}}{{R_n^b{\tau ^b}}} + t_n^{b,trip}}}{{\dfrac{{D_{n,k}^{{u_b}}C_{n,k}^{{u_b}}K_n^b}}{{{F^b}}} + \dfrac{{D_{n,k}^{{u_b}}{K^{{u_b},c}}}}{{R_n^b{\tau ^b}}}}}, \nonumber \\
&\kappa _{n,k}^{{u_b},cloud} = \dfrac{{\dfrac{{(1 - \alpha _{n,k}^{{u_b},cloud})D_{n,k}^{{u_b}}C_{n,k}^{{u_b}}K_n^b}}{{{F^b}}} - t_n^{b,trip}}}{{\dfrac{{D_{n,k}^{{u_b}}C_{n,k}^{{u_b}}K_n^b}}{{{F^b}}} + \dfrac{{D_{n,k}^{{u_b}}{K^{{u_b},c}}}}{{R_n^b{\tau ^b}}}}}. \nonumber
\end{align}
Then, the task completion delay of task $Q_{n,k}^{{u_b}}$ is obtained as
\begin{align}\label{eq:BS_user_task_completion_delay_edge_cloud}
t_{n,k}^{{u_b},total,cloud} = t_{n,k}^{{u_b},u,cloud} = t_{n,k}^{{u_b},b,cloud} = t_{n,k}^{{u_b},c,cloud} = \frac{{\alpha _{n,k}^{{u_b},cloud}D_{n,k}^{{u_b}}C_{n,k}^{{u_b}}}}{{f_{n,k}^{{u_b}}}}.
\end{align}

\subsection{Optimal Task Partitioning Strategy of TUEs}\label{sec:4_5}
Based on the analysis above, the task completion delay of TUEs is given by (\ref{eq:BS_user_task_completion_delay_edge_cloud}) if offloaded to the cloud, or given by (\ref{eq:BS_user_task_completion_delay_local_edge}) if not offloaded to the cloud. For task $Q_{n,k}^{{u_b}}$, we can see that as long as the condition $t_{n,k}^{b,comp,edge} > t_n^{b,trip}$ is satisfied, the task completion delay in (\ref{eq:BS_user_task_completion_delay_edge_cloud}) is always lower than the task completion delay in (\ref{eq:BS_user_task_completion_delay_local_edge}).
To guarantee the fairness among users, we consider that task $Q_{n,k}^{{u_b}}$ will be offloaded to the cloud if the condition $t_{n,k}^{b,comp,edge} > t_n^{b,trip}$ is satisfied, and vice versa. Then, the number of TUEs offloaded to the cloud can be calculated by
\begin{align}\label{eq:BS_user_cloud_number_calculation}
{K^{{u_b},c}} = \sum\limits_{n = 1}^N {\sum\limits_{k = 1}^{{K^b_n}} {\Gamma(t_{n,k}^{b,comp,edge} > t_n^{trip})} }.
\end{align}

Finally, we turn back to subproblem ${{\cal P}^{{u_b},delay}}$, and propose Algorithm \ref{alg:optimal_task_partitioning_BS} to calculated the optimal task partitioning strategy of TUEs. In Algorithm \ref{alg:optimal_task_partitioning_BS}, we first calculate the local-edge task partitioning strategy $\alpha _{n,k}^{{u_b},edge},\beta _{n,k}^{{u_b},edge}$ for each task, based on which the task completion delay $t_{n,k}^{{u_b},total,edge}$ and the edge computing delay $t_{n,k}^{b,comp,edge}$ can also be calculated. Then, we determine the task partitioning strategy of each task by comparing the edge computing delay $t_{n,k}^{b,comp,edge}$ with the propagation delay $t_n^{b,trip}$. If $t_{n,k}^{b,comp,edge} \le t_n^{b,trip}$, the task partitioning strategy of task $Q_{n,k}^{{u_b}}$ is given by $\alpha _{n,k}^{{u_b}*} = \alpha _{n,k}^{{u_b},edge},\beta _{n,k}^{{u_b}*} = \beta _{n,k}^{{u_b},edge},\kappa _{n,k}^{{u_b}*} = 0$. If $t_{n,k}^{b,comp,edge} > t_n^{b,trip}$, the task partitioning strategy of task $Q_{n,k}^{{u_b}}$ is given by $\alpha _{n,k}^{{u_b}*} = \alpha _{n,k}^{{u_b},cloud},\beta _{n,k}^{{u_b}*} = \beta _{n,k}^{{u_b},cloud},\kappa _{n,k}^{{u_b}*} = \kappa _{n,k}^{{u_b},cloud}$.
Note that the complexity for calculating the task partitioning strategy and task completion delay is $\mathcal{O}(1)$ in one iteration. The overall complexity for executing Algorithm \ref{alg:optimal_task_partitioning_BS} is given by $\mathcal{O}({K^\textrm{B}})$.

\begin{algorithm}[hb]
\caption{Optimal Task Partitioning Strategy of TUEs} \label{alg:optimal_task_partitioning_BS}
\begin{algorithmic}[1]
\begin{spacing}{1.0}
\STATE Initialize ${K^{{u_b},c}} = 0$
\FOR{$n = 1$ to $N$}
  \FOR{$k= 1$ to $K_n^b$}
    \STATE Calculate $\alpha _{n,k}^{{u_b},edge},\beta _{n,k}^{{u_b},edge}$ according to (\ref{eq:BS_user_task_partitioning_ratio_local_edge})
    \STATE Calculate $t_{n,k}^{{u_b},total,edge},t_{n,k}^{b,comp,edge}$ according to (\ref{eq:BS_user_task_completion_delay_local_edge})
    \STATE Update ${K^{{u_b},c}} = {K^{{u_b},c}} + \Gamma(t_{n,k}^{b,comp,edge} > t_n^{b,trip})$
  \ENDFOR
\ENDFOR
\FOR{$n = 1$ to $N$}
  \FOR{$k= 1$ to $K_n^b$}
    \IF{$t_{n,k}^{b,comp,edge} \le t_n^{b,trip}$}
      \STATE Update the task partitioning strategy by $\alpha _{n,k}^{{u_b}*} = \alpha _{n,k}^{{u_b},edge},\beta _{n,k}^{{u_b}*} = \beta _{n,k}^{{u_b},edge},\kappa _{n,k}^{{u_b}*} = 0$
      \STATE Update the task completion delay by $t_{n,k}^{{u_b},total*} = t_{n,k}^{{u_b},total,edge}$
    \ELSE
      \STATE Update $\alpha _{n,k}^{{u_b}*} = \alpha _{n,k}^{{u_b},cloud},\beta _{n,k}^{{u_b}*} = \beta _{n,k}^{{u_b},cloud},\kappa _{n,k}^{{u_b}*} = \kappa _{n,k}^{{u_b},cloud}$ according to (\ref{eq:BS_user_task_partitioning_ratio_edge_cloud})
      \STATE Calculate the task completion delay by $t_{n,k}^{{u_b},total*} = t_{n,k}^{{u_b},total,cloud}$ according to (\ref{eq:BS_user_task_completion_delay_edge_cloud})
    \ENDIF
  \ENDFOR
\ENDFOR
\vspace{-5mm}
\end{spacing}
\end{algorithmic}
\end{algorithm}

\section{Task Partitioning Strategy of SUE}\label{sec:5}

\subsection{Problem Decomposition}
As analyzed in Section \ref{sec:4_1}, if we keep the subframe allocation strategy ${{{\bm{\tau }}^{{u_b}}},{{\bm{\theta }}^{{u_b},rad}},{\tau ^b},{\tau ^{{u_s}}},{\theta ^{{u_s},rad}}}$ fixed, the task partitioning strategy of SUEs can be obtained by solving the subproblem as follows:
\begin{align}\label{eq:subproblem_satellite_user_all}
&{{\cal P}^{{u_s},delay}}:
\mathop {\min }\limits_{{{\bm{\alpha }}^{{u_s}}},{{\bm{\kappa }}^{{u_s}}}} {t^{{u_s},total}} = \sum\limits_{k = 1}^{{K^s}} {t_k^{{u_s},total}},  \\
&\hspace{15mm} s.t.\;\alpha _k^{{u_s}} + \kappa _k^{{u_s}} = 1,\forall k, \nonumber \\
&\hspace{15mm}\;\;\;\;\;\, \alpha _k^{{u_s}},\kappa _k^{{u_s}} \ge 0,\forall k. \nonumber
\end{align}

Similarly, since the completion delay of task $Q_k^{{u_s}}$ in (\ref{eq:total_delay_satellite_user}) is a discontinuous function, the task partitioning strategy of SUEs need to be analyzed differently based on whether it is offloaded to the cloud.
For any task $Q_k^{{u_s}}$, if it is not offloaded to the cloud, we have $\alpha _k^{{u_s},local} = 1,\kappa _k^{{u_s},local} = 0$. The task completion delay is given by
\begin{align}\label{eq:satellite_user_task_completion_delay_local}
t_k^{{u_s},total,local} = t_k^{{u_s},comp,local} = \frac{{D_k^{{u_s}}C_k^{{u_s}}}}{{f_k^{{u_s}}}}.
\end{align}
Then, by comparing the local computing delay $t_k^{{u_s},comp,local}$ with the propagation delay $t_k^{{u_s},trip}$, the following two cases will be discussed in Section \ref{sec:5_2} and Section \ref{sec:5_3} separately.

\subsection{Optimal Strategy Discussion: When $t_k^{{u_s},comp,local} \le t_k^{{u_s},trip}$}\label{sec:5_2}
Assume that task $Q_k^{{u_s}}$ is offloaded to the cloud, and the total amount of data executed at the cloud is ${\kappa _k^{{u_s}}D_k^{{u_s}}}$. The task completion delay of task $Q_k^{{u_s}}$ is calculated by $t_k^{{u_s},total} = \max \{ t_k^{{u_s},u},t_k^{{u_s},c}\}$. Since $t_k^{{u_s},comp,local} \le t_k^{{u_s},trip}$, the following inequality holds for task $Q_k^{{u_s}}$ when $\kappa _k^{{u_s}} > 0$.
\begin{align}\label{eq:satellite_user_case1_total}
t_k^{{u_s},total} \ge t_k^{{u_s},c} = \frac{{\kappa _k^{{u_s}}D_k^{{u_s}}{K^{{u_s},c}}}}{{R_k^{{u_s}}{\tau ^{{u_s}}}}} + t_k^{{u_s},trip} > t_k^{{u_s},comp,local} = t_k^{{u_s},total,local}.
\end{align}

When $t_k^{{u_s},comp,local} \le t_k^{{u_s},trip}$, we prove that $t_k^{{u_s},total} > t_k^{{u_s},total,local}$ holds for any $\kappa _k^{{u_s}} > 0$. Thus task $Q_k^{{u_s}}$ should not be offloaded to the cloud. The optimal task partitioning strategy is given by $\alpha _k^{{u_s}*} = 1,\kappa _k^{{u_s}*} = 0$.

\subsection{Optimal Strategy Discussion: When $t_k^{{u_s},comp,local} > t_k^{{u_s},trip}$}\label{sec:5_3}
When $t_k^{{u_s},comp,local} > t_k^{{u_s},trip}$, there will always exist a task partitioning strategy with $\kappa _k^{{u_s}} > 0$, for which the task completion delay of task $Q_k^{{u_s}}$ satisfies $t_k^{{u_s},total} < t_k^{{u_s},total,local}$. This conclusion can be directly derived based on the expressions of $t_k^{{u_s},total}$ and $t_k^{{u_s},total,local}$ in (\ref{eq:total_delay_satellite_user}) and (\ref{eq:satellite_user_task_completion_delay_local}).
In this case, lower task completion delay can be achieved by offloading the task to the cloud.

Similar to the analysis of TUEs, the task completion delay is minimized when $t_k^{{u_s},u} = t_k^{{u_s},c}$. By solving the boundary condition, the local-cloud task partitioning strategy of task $Q_k^{{u_s}}$ is derived
\begin{align}\label{eq:satellite_user_task_partitioning_ratio_local_cloud}
&\alpha _k^{{u_s},cloud} = \dfrac{{\dfrac{{D_k^{{u_s}}{K^{{u_s},c}}}}{{R_k^{{u_s}}{\tau ^{{u_s}}}}} + t_k^{{u_s},trip}}}{{\dfrac{{D_k^{{u_s}}{K^{{u_s},c}}}}{{R_k^{{u_s}}{\tau ^{{u_s}}}}} + \dfrac{{D_k^{{u_s}}C_k^{{u_s}}}}{{f_k^{{u_s}}}}}}, \\
&\kappa _k^{{u_s},cloud} = \dfrac{{\dfrac{{D_k^{{u_s}}C_k^{{u_s}}}}{{f_k^{{u_s}}}} - t_k^{{u_s},trip}}}{{\dfrac{{D_k^{{u_s}}{K^{{u_s},c}}}}{{R_k^{{u_s}}{\tau ^{{u_s}}}}} + \dfrac{{D_k^{{u_s}}C_k^{{u_s}}}}{{f_k^{{u_s}}}}}}. \nonumber
\end{align}
Then, the task completion delay of task $Q_k^{{u_s}}$ is obtained as
\begin{align}\label{eq:satellite_user_task_completion_delay_local_cloud}
t_k^{{u_s},total,cloud} = t_k^{{u_s},u,cloud} = t_k^{{u_s},c,cloud} = \frac{{\alpha _k^{{u_s},cloud}D_k^{{u_s}}C_k^{{u_s}}}}{{f_k^{{u_s}}}}.
\end{align}

\subsection{Optimal Task Partitioning Strategy of SUEs}\label{sec:5_4}
In the analysis above, we calculate the task partitioning strategy and the task completion delay for the two different cases. Similarly, to guarantee the fairness among users, we consider that task $Q_k^{{u_s}}$ will be offloaded to the cloud if the condition $t_k^{{u_s},comp,local} > t_k^{{u_s},trip}$ is satisfied, and vice versa. Then, the number of SUEs offloaded to the cloud can be calculated by
\begin{align}\label{eq:satellite_user_cloud_number_calculation}
{K^{{u_s},c}}  = \sum\limits_{k = 1}^{{K^s}} {\Gamma(t_k^{{u_s},comp,local} > t_k^{{u_s},trip})}.
\end{align}

Finally, we turn back to subproblem ${{\cal P}^{{u_s},delay}}$, and propose Algorithm \ref{alg:optimal_task_partitioning_satellite} to calculated the optimal task partitioning strategy of SUEs. In Algorithm \ref{alg:optimal_task_partitioning_satellite}, we first calculate the task completion delay $t_k^{{u_s},total,local}$ and the local computing delay $t_{n,k}^{{u_s},comp,local}$ for each task. Then, we determine the task partitioning strategy of each task by comparing the local computing delay $t_{n,k}^{{u_s},comp,local}$  with the propagation delay $t_k^{{u_s},trip}$. If $t_k^{{u_s},comp,local} \le t_k^{{u_s},trip}$, the task partitioning strategy of task $Q_k^{{u_s}}$ is given by $\alpha _k^{{u_s}*} = 1,\kappa _k^{{u_s}*} = 0$. If $t_k^{{u_s},comp,local} > t_k^{{u_s},trip}$, the task partitioning strategy of task $Q_k^{{u_s}}$ is given by $\alpha _k^{{u_s}*} = \alpha _k^{{u_s},cloud},\kappa _k^{{u_s}*} = \kappa _k^{{u_s},cloud}$.
Note that the complexity for calculating the task partitioning strategy and task completion delay is $\mathcal{O}(1)$ in one iteration. The overall complexity for executing Algorithm \ref{alg:optimal_task_partitioning_satellite} is given by $\mathcal{O}({K^\textrm{S}})$.

\begin{algorithm}[thb]
\caption{Optimal Task Partitioning Strategy of SUEs} \label{alg:optimal_task_partitioning_satellite}
\begin{algorithmic}[1]
\begin{spacing}{1.0}
\STATE Initialize ${K^{{u_s},c}} = 0$
\FOR{$k= 1$ to $K^\textrm{S}$}
  \STATE Calculate $t_k^{{u_s},total,local},  t_k^{{u_s},comp,local}$ according to (\ref{eq:satellite_user_task_completion_delay_local})
  \STATE Update ${K^{{u_s},c}} = {K^{{u_s},c}} + \Gamma(t_k^{{u_s},comp,local} > t_k^{{u_s},trip})$
\ENDFOR
\FOR{$k= 1$ to $K^\textrm{S}$}
  \IF{$t_k^{{u_s},comp,local} \le t_k^{{u_s},trip}$}
    \STATE Update the task partitioning strategy by $\alpha _k^{{u_s}*} = 1,\kappa _k^{{u_s}*} = 0$
    \STATE Update the task completion delay by $t_k^{{u_s},total*} = t_k^{{u_s},total,local}$
  \ELSE
    \STATE Update $\alpha _k^{{u_s}*} = \alpha _k^{{u_s},cloud},\kappa _k^{{u_s}*} = \kappa _k^{{u_s},cloud}$ according to (\ref{eq:satellite_user_task_partitioning_ratio_local_cloud})
    \STATE Calculate the task completion delay by $t_k^{{u_s},total*} = t_k^{{u_s},total,cloud}$ according to (\ref{eq:satellite_user_task_completion_delay_local_cloud})
  \ENDIF
\ENDFOR
\vspace{-5mm}
\end{spacing}
\end{algorithmic}
\end{algorithm}
\vspace{-0mm}

\section{Joint Subframe Allocation and Task Partitioning Strategy}\label{sec:6}
For any given subframe allocation strategy, which is represented by ${{{\bm{\tau }}^{{u_b}}},{{\bm{\theta }}^{{u_b},rad}},{\tau ^b},{\tau ^{{u_s}}},{\theta ^{{u_s},rad}}}$, the task partitioning strategies of TUEs and SUEs can be calculated based on Algorithm \ref{alg:optimal_task_partitioning_BS} and Algorithm \ref{alg:optimal_task_partitioning_satellite}. In this perspective, the objective function of problem $\cal P$ can be simplified as the function of the subframe allocation strategy ${U^{total}}({{\bm{\tau }}^{{u_b}}},{{\bm{\theta }}^{{u_b},rad}},{\tau ^b},{\tau ^{{u_s}}},{\theta ^{{u_s},rad}})$. Then, the original optimization problem $\cal P$ in (\ref{eq:optimization_problem}) is equivalent to
\begin{align}\label{eq:optimization_problem_equivalent}
&{{\cal P}^{equ}}:
\mathop {\max }\limits_{{{\bm{\tau }}^{{u_b}}},{{\bm{\theta }}^{{u_b},rad}},{\tau ^b},{\tau ^{{u_s}}},{\theta ^{{u_s},rad}}} {U^{total}}({{\bm{\tau }}^{{u_b}}},{{\bm{\theta }}^{{u_b},rad}},{\tau ^b},{\tau ^{{u_s}}},{\theta ^{{u_s},rad}}),  \\
&\hspace{9mm} s.t.\;\tau _n^{{u_b}} + \theta _n^{{u_b},rad} = 1,\tau _n^{{u_b}},\theta _n^{{u_b},rad} > 0,\forall n, \nonumber \\
&\hspace{9mm} \;\;\;\;\;\; {\tau ^b} + {\tau ^{{u_s}}} + {\theta ^{{u_s},rad}} = 1,{\tau ^b},{\tau ^{{u_s}}},{\theta ^{{u_s},rad}} > 0. \nonumber
\end{align}

\subsection{Particle Swarm Optimization}

With respect to the equivalently transformed problem ${{\cal P}^{equ}}$, the particle swarm optimization (PSO) method is then utilized to calculate the joint subframe allocation and task partitioning strategy. Derived from the social interactions among birds, the PSO method is efficient for solving complex optimization problems with widespread applications \cite{Xie2020Energy7734,Luo2022Minimizing2897}. The effectiveness of the PSO algorithm has been verified by plenty of works and experiments \cite{Chantre2018Multi2304,Huang2022TMA1652}.

When solving problem ${{\cal P}^{equ}}$ by the PSO method, the possible solutions of problem ${{\cal P}^{equ}}$ are modeled as particles in the search space. For particle $l$, the position ${{\bf{x}}_l}$ is the projection of the subframe allocation strategy onto the search space, given by
\begin{align}\label{eq:position_vector}
&{{\bf{x}}_l} = [{\bf{x}}_l^{{\tau ^{{u_b}}}},{\bf{x}}_l^{{\theta ^{{u_b},rad}}},x_l^{{\tau ^b}},x_l^{{\tau ^{{u_s}}}},x_l^{{\theta ^{{u_s},rad}}}],\\
&{\bf{x}}_l^{{\tau ^{{u_b}}}} = [x_{l,1}^{{\tau ^{{u_b}}}},...,x_{l,n}^{{\tau ^{{u_b}}}},...,x_{l,N}^{{\tau ^{{u_b}}}}], \nonumber \\
&{\bf{x}}_l^{{\theta ^{{u_b}}}} = [x_{l,1}^{{\theta ^{{u_b},rad}}},...,x_{l,n}^{{\theta ^{{u_b},rad}}},...,x_{l,N}^{{\theta ^{{u_b},rad}}}], \nonumber
\end{align}
where ${\bf{x}}_l^{{\tau ^{{u_b}}}}$ is the projection of the communication subframe allocation strategy $\bm{\tau} ^{{u_b}}$, and ${\bf{x}}_l^{{\theta ^{{u_b}}}}$ is the projection of the sensing subframe allocation strategy $\bm{\theta} ^{{u_b}}$. Also, $x_l^{{\tau ^b}},x_l^{{\tau ^{{u_s}}}},x_l^{{\theta ^{{u_s},rad}}}$ are the projection of the subframe allocation strategy ${{\tau ^b},{\tau ^{{u_s}}},{\theta ^{{u_s},rad}}}$.

To characterize the movement of particles in the search space, each particle $l$ is also assigned a velocity attribute ${{\bf{v}}_l}$, given by
\begin{align}\label{eq:velocity_vector}
&{{\bf{v}}_l} = [{\bf{v}}_l^{{\tau ^{{u_b}}}},{\bf{v}}_l^{{\theta ^{{u_b},rad}}},v_l^{{\tau ^b}},v_l^{{\tau ^{{u_s}}}},v_l^{{\theta ^{{u_s},rad}}}], \\
&{\bf{v}}_l^{{\tau ^{{u_b}}}} = [v_{l,1}^{{\tau ^{{u_b}}}},...,v_{l,n}^{{\tau ^{{u_b}}}},...,v_{l,N}^{{\tau ^{{u_b}}}}], \nonumber\\
&{\bf{v}}_l^{{\theta ^{{u_b}}}} = [v_{l,1}^{{\theta ^{{u_b},rad}}},...,v_{l,n}^{{\theta ^{{u_b},rad}}},...,v_{l,N}^{{\theta ^{{u_b},rad}}}]. \nonumber
\end{align}
The velocity attribute ${{\bf{v}}_l}$ characterizes the moving speed in each dimension of the search space, corresponding to the variation rate of the subframe allocation strategy.

In order to find the optimal solution, each particle will move iteratively in the search space based on its own experience, and will also learn from other particles in a cooperative manner. For particle $l$, the updating process in the $\delta$-th iteration is defined as follows:
\begin{align}\label{eq:particle_update}
&{{\bf{v}}_l}(\delta  + 1) = \omega {{\bf{v}}_l}(\delta ) + {a_1}{z_1}[{\bm{\rho}} _l^{best}(\delta ) - {{\bf{x}}_l}(\delta )]\\
&\hspace{25.2mm} + {a_2}{z_2}[{{\bm{\rho}} ^{gbest}}(\delta ) - {{\bf{x}}_l}(\delta )], \nonumber\\
&{{\bf{x}}_l}(\delta  + 1) = {{\bf{x}}_l}(\delta ) + {{\bf{v}}_l}(\delta  + 1), \nonumber
\end{align}
in which $\omega$ denotes the inertia weight, $a_1, a_2$ denote the learning factors, and ${z_1},{z_2}$ denote two random functions obeying uniform distribution in $[0,1]$. The constriction coefficient can also be utilized to guarantee the convergence of population and prevent explosion of the velocity \cite{Clerc2002particle58}. More detailed convergence analysis can be found in \cite{Poli2007Particle33}. Also, ${\bm{\rho}} _l^{best}$ denotes the personal best solution that particle $l$ has experienced during the iteration, and ${\bm{\rho}}^{gbest}$ denotes the global best solution selected from all personal best solutions.

In the iteration, the fitness of position ${{\bf{x}}_l}$ is evaluated by the value of the objective function, represented by ${U^{total}}({{\bf{x}}_l})$.
Then, the updating process of the personal best solution and the global best solution is defined as follows:
\begin{align}\label{eq:best_solution_update}
&{\bm{\rho}} _l^{best}(\delta  + 1) = \left\{ {\begin{array}{*{20}{l}}
{{\bm{\rho}} _l^{best}(\delta ),
\hspace{2.8mm} {U^{total}}({{\bf{x}}_l}(\delta  + 1)) \le {U^{total}}({\bm{\rho}} _l^{best}(\delta )),}\\
{{{\bf{x}}_l}(\delta  + 1),{U^{total}}({{\bf{x}}_l}(\delta  + 1)) > {U^{total}}({\bm{\rho}} _l^{best}(\delta )),}
\end{array}} \right. \\
&{{\bm{\rho}} ^{gbest}}(\delta  + 1) = \arg \mathop {\max }\limits_{{\bm{\rho}} _l^{best}} {U^{total}}({\bm{\rho}} _l^{best}(\delta  + 1)). \nonumber
\end{align}

\subsection{Joint Subframe Allocation and Task Partitioning Strategy}

Based on the structure of the PSO method discussed above, the joint subframe allocation and task partitioning strategy is illustrated in Algorithm \ref{alg:joint}.
First, we generate a group of particles of population $L$ within the search space. During the iteration process, each particle will learn and move according to the updating function in (\ref{eq:particle_update}). Note that the particles should always be maintained within the search space, which is restricted by the constraints in (\ref{eq:optimization_problem_equivalent}). For the new position obtained, the fitness is evaluated by calculating the value of the objective function ${U^{total}}({{\bf{x}}_l}(\delta  + 1))$, based on which the personal best solution ${\bm{\rho}} _l^{best}(\delta  + 1)$ and the global best solution ${{\bm{\rho}} ^{gbest}}(\delta  + 1)$ will be updated. Finally, we set the subframe allocation strategy as ${{\bm{\tau }}^{{u_b}*}},{{\bm{\theta }}^{{u_b},rad*}},{\tau ^{b*}},{\tau ^{{u_s}*}},{\theta ^{{u_s},rad*}} = {{\bm{\rho}} ^{gbest}}({\delta _{max}})$, while the task partitioning strategy of TUEs and SUEs can be calculated with Algorithm \ref{alg:optimal_task_partitioning_BS} and Algorithm \ref{alg:optimal_task_partitioning_satellite}.
The complexity of Algorithm \ref{alg:joint} mainly comes from calculating the task partitioning strategy with Algorithm \ref{alg:optimal_task_partitioning_BS} and Algorithm \ref{alg:optimal_task_partitioning_satellite} in each iteration. The overall complexity for executing Algorithm \ref{alg:joint} is given by $\mathcal{O}(\delta_{max} L K)$.

\begin{algorithm}[hb]
\caption{Joint Subframe Allocation and Task Partitioning Strategy} \label{alg:joint}
\begin{algorithmic}[1]
\begin{spacing}{1}
\STATE Generate the particle swarm of population $L$
\FOR{$\delta = 1,2,..., \delta_{max}$}
  \FOR {particle $l=1,2,...,L$}
    \STATE Calculate the velocity ${{\bf{v}}_{l}}(\delta  + 1)$ and position ${{\bf{x}}_{l}}(\delta  + 1)$ of particle $l$ according to (\ref{eq:particle_update})
    \IF{${{\bf{x}}_{l}}(\delta  + 1)$ is out of the search space}
      \STATE Set ${{\bf{x}}_{l}}(\delta  + 1) = {\bf{x}}_{l}^{bound}(\delta  + 1)$
    \ENDIF
    \STATE Set the subframe allocation strategy as ${{\bm{\tau }}^{{u_b}}},{{\bm{\theta }}^{{u_b},rad}},{\tau ^{b}},{\tau ^{{u_s}}},{\theta ^{{u_s},rad}} = {{\bf{x}}_{l}}(\delta  + 1)$
    \STATE Set TUE task partitioning strategy by Algorithm \ref{alg:optimal_task_partitioning_BS}
    \STATE Set SUE task partitioning strategy by Algorithm \ref{alg:optimal_task_partitioning_satellite}
    \STATE Evaluate the fitness of the new position by calculating ${U^{total}}({{\bf{x}}_l}(\delta  + 1))$
    \STATE Calculate the personal best solution ${\bm{\rho}} _l^{best}(\delta  + 1)$ by (\ref{eq:best_solution_update})
  \ENDFOR
  \STATE Calculate the global best solution ${{\bm{\rho}} ^{gbest}}(\delta  + 1)$ by (\ref{eq:best_solution_update})
\ENDFOR
\STATE Set the subframe allocation strategy as ${{\bm{\tau }}^{{u_b}*}},{{\bm{\theta }}^{{u_b},rad*}},{\tau ^{b*}},{\tau ^{{u_s}*}},{\theta ^{{u_s},rad*}} = {{\bm{\rho}} ^{gbest}}({\delta _{max}})$
\STATE Set TUE task partitioning strategy by Algorithm \ref{alg:optimal_task_partitioning_BS}
\STATE Set SUE task partitioning strategy by Algorithm \ref{alg:optimal_task_partitioning_satellite}
\vspace{-5mm}
\end{spacing}
\end{algorithmic}
\end{algorithm}
\vspace{-0mm}

\section{Simulation Results}\label{sec:7}

\subsection{Simulation Setup}

In the simulation, the following settings are applied as the default value of the parameters. For the network model, the satellite is considered to be on the orbit of 8,000 km, the number of SUEs is set to be $K^\textrm{S}=$ 30, and the number of the BSs is set to be $N=$ 10. For each BS, the number of TUEs is set as $K_n^b =$ 5, which are randomly distributed within the cell radius of 500 m.
For the communication and sensing model, the carrier frequency of 28 GHz is considered for both the BS and the satellite. The terrestrial communication channel is modeled referring to \cite{Zhang2022Time2206}, the satellite communication channel is modeled referring to \cite{Jia2021Joint1147,Chr2015Multicast4695}, and the sensing channel is modeled referring to \cite{Zhang2021Design1500}. For the computing model, the data amount of each computation task is generated by uniform distribution with $D_{n,k}^{{u_b}}\in$ [100,~900] Kb and $D_{k}^{{u_s}}\in$ [100,~900] Kb, while the computation workload is set as $C_{n,k}^{{u_b}} = $ 1,000 cycle/bit and $C_{k}^{{u_s}} = $ 1,000 cycle/bit. Other key parameters are summarized in Table \ref{tab:parameter}.

\begin{table}[t]\footnotesize
  \centering
  \begin{spacing}{1.3}
  \caption{Simulation parameters.}\label{tab:parameter}
  \begin{tabular}{|m{4.5cm}<{\centering}|m{2.5cm}<{\centering}|}
  \hline
  \textbf{Parameter}  &  \textbf{Value}  \\
  \hline
  Terrestrial bandwidth $B_1$    &  20 MHz  \\
  \hline
  Satellite   bandwidth $B_2$    &  40 MHz \\
  \hline
  Terrestrial frame length  $T_1$  &  10 ms \\
  \hline
  Satellite   frame length  $T_2$  &  10 ms \\
  \hline
  User transmission power $p_{n,k}^{u_b}$	& 30 dBm \\
  \hline
  User transmission power $p_{k}^{u_s}$	    & 30 dBm \\
  \hline
  BS transmission power $p_{n}^b$	& 40 dBm     \\
  \hline
  Noise power density   $N_0$       &-174 dBm/Hz \\
  \hline
  Terrestrial antenna gain $G^t$/$G^r$ & 18 dB \\
  \hline
  Satellite   antenna gain $G^t$/$G^r$ & 40 dB \\
  \hline
  Task data amount $D_{n,k}^{u_b}$	& [100,~900] Kb \\
  \hline
  Task data amount $D_{k}^{u_s}$	    & [100,~900] Kb \\
  \hline
  Task  Workload $C_{n,k}^{u_b}$ & 1,000 cycle/bit \\
  \hline
  Task  Workload $C_{k}^{u_s}$   & 1,000 cycle/bit \\
  \hline
  User computing frequency $f_{n,k}^{u_b}$	& $5 \times 10^8$ cycle/s \\
  \hline
  User computing frequency $f_{k}^{u_s}$	& $5 \times 10^8$ cycle/s \\
  \hline
  BS computing frequency $F_b$	& $5 \times 10^9$ cycle/s \\
  \hline
  Propagation delay $t^{cloud}$ & 100 ms \\
  \hline
  \end{tabular}
  \end{spacing}
\end{table}

The following optimization strategies are implemented for comparisons in the simulation:
\begin{enumerate}[~~1.]
\item \textbf{JSATPS:} This is the joint subframe allocation and task partitioning strategy (JSATPS) given by the proposed Algorithm \ref{alg:joint}.
\item \textbf{Greedy-OTPS:} The subframe allocation strategy is obtained by solving problem ${\cal P}^{equ}$ in a greedy manner with maximum additional utility \cite{Zhao2022Radio8675}, while the task partitioning strategies are calculated by the proposed Algorithm \ref{alg:optimal_task_partitioning_BS} and Algorithm \ref{alg:optimal_task_partitioning_satellite}.
\item \textbf{Greedy-Equal:} The subframe allocation strategy is obtained in a greedy manner with maximum additional utility, while the computation tasks are divided with equal size \cite{Liu2020Post3170}.
\item \textbf{Exhaustive:} The optimal solution is obtained by exhaustive search.
\end{enumerate}

\subsection{Performance Analysis}

\begin{figure*} [t]
\begin{center}
\begin{minipage}[t]{0.45 \linewidth}
\centering
\includegraphics*[width=3in]{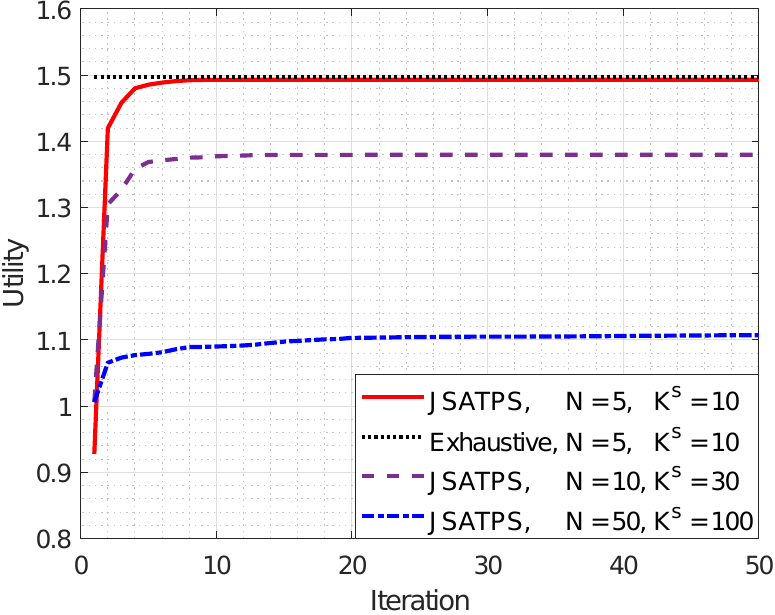}
\captionsetup{font={footnotesize}}
\caption{Convergence performance of Algorithm \ref{alg:joint}.}\label{fig:iteration}
\end{minipage}
\begin{minipage}[t]{0.45 \linewidth}
\centering
\includegraphics*[width=3in]{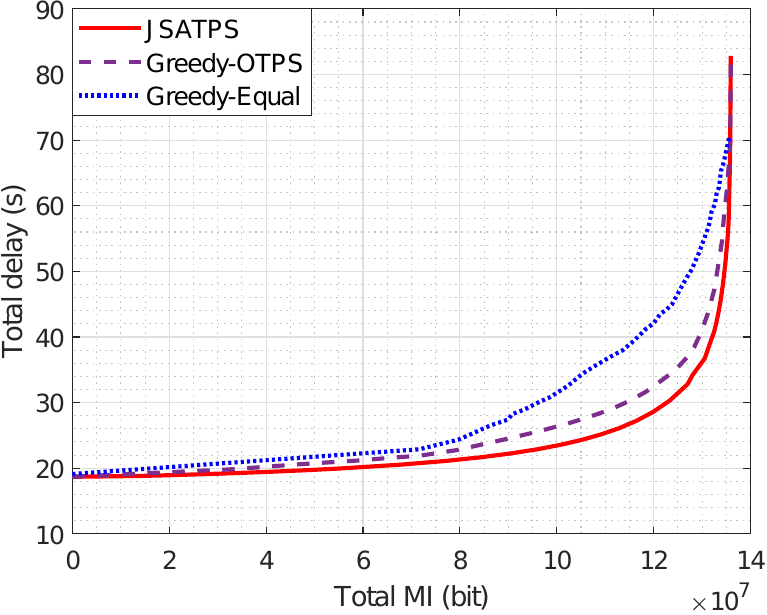}
\captionsetup{font={footnotesize}}
\caption{Trade-off between radar MI and task completion delay.}\label{fig:pareto}
\end{minipage}
\end{center}
\vspace{-0mm}
\end{figure*}

Fig. \ref{fig:iteration} gives the convergence performance of Algorithm \ref{alg:joint} when implementing the proposed JSATPS strategy for different network scales. Results show that Algorithm \ref{alg:joint} can achieve stable convergence performance under different network scales. Although the convergence speed may be slower for large scales, the utility function can generally converge within 50 iterations, validating the feasibility of the proposed strategy.
Also, to give a intuitive evaluation of the proposed strategy, we calculate the optimal solution by exhaustive search for the small case of $N =$ 5, $K^\textrm{S} =$ 10. It can be seen that the proposed JSATPS strategy can achieve near-optimal performance compared with the result obtained by exhaustive search, for which the performance loss is less than 1\%. The reliability of the proposed strategy is also validated.

Fig. \ref{fig:pareto} illustrates the trade-off between the radar MI and the task completion delay for the proposed strategy and the comparison strategies, in which the total MI refers to the total radar MI $I^{total}$, and the total delay refers to the total task completion delay $t^{total}$. As discussed in the problem formulation, we introduce the Cobb-Douglas utility function to characterize the multi-dimensional resource demands and multi-objective performance optimization of the network. By adjusting the weighting metric $\eta$, the Pareto optimal solutions can be obtained along the Pareto frontier \cite{Xu2023UAV}.
It can be seen that the Pareto frontier of the proposed JSATPS strategy outperforms the Greedy-OTPS strategy and the Greedy-Equal strategy, achieving larger radar MI and lower task completion delay. The superiority over the comparison strategies validates the effectiveness of the proposed strategy.
Also, since the two optimization objectives of the network conflict with each other, increasing the radar MI will lead to the increase of the task completion delay when moving along the Pareto frontier. For the JSATPS strategy, when the radar MI increases from 0 to $\mathrm{1.36 \times 10^8}$ bits, the task completion delay will increase from 18.7 s to 82.8 s. However, it can be seen that the marginal effect exists for the trade-off between the radar MI and the task completion delay. For the JSATPS strategy, if we increase the radar MI by 100 \% from $\mathrm{6 \times 10^7}$ bits to $\mathrm{1.2 \times 10^8}$ bits, more subframe resources will be allocated for sensing, and the task completion delay will increase by 40 \% from 20 s to 28 s. However, if we further allocate all the subframe resources for sensing, the radar MI only increases by 13 \% from $\mathrm{1.2 \times 10^8}$ bits to $\mathrm{1.36 \times 10^8}$ bits, while the task completion delay increases dramatically by near 200 \% from 28 s to 82.8 s. Due to the marginal effect, the performance gain of the radar MI is much more smaller than the performance loss of the task completion delay in this case. Similarly, in order to reduce the task completion delay by only 6.5 \% from 20 s to 18.7 s, all the subframe resources need to be allocated for communication, and the radar MI will decrease from $\mathrm{6 \times 10^7}$ bits to zero in this case. It can be seen that compared with the edge points on the Pareto frontier, it is more efficient to select the middle points when optimizing the network performance.

\begin{figure*} [tbp]
\begin{center}
\begin{minipage}[t]{0.45 \linewidth}
\centering
\includegraphics*[width=3in]{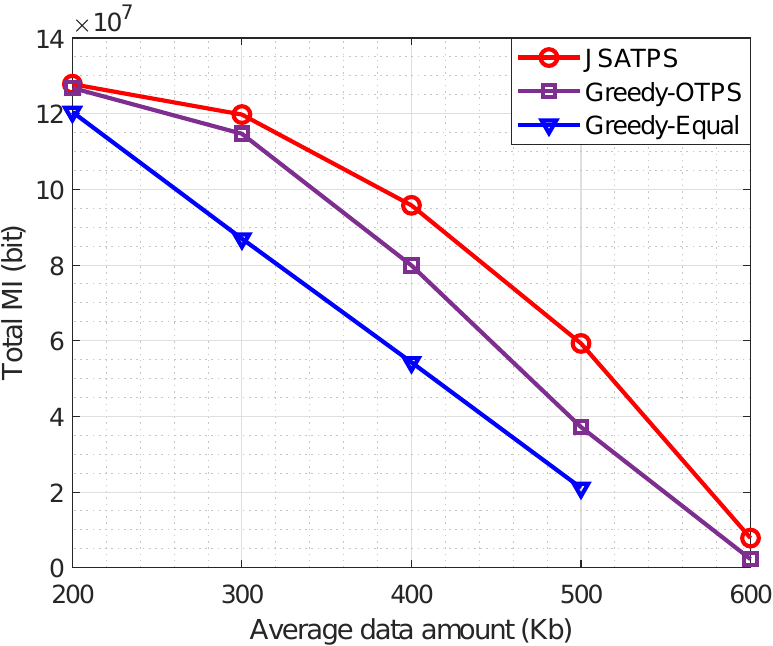}
\centerline{\footnotesize (a) $t^{total}= 20$ s.}
\end{minipage}
\begin{minipage}[t]{0.45 \linewidth}
\centering
\includegraphics*[width=3in]{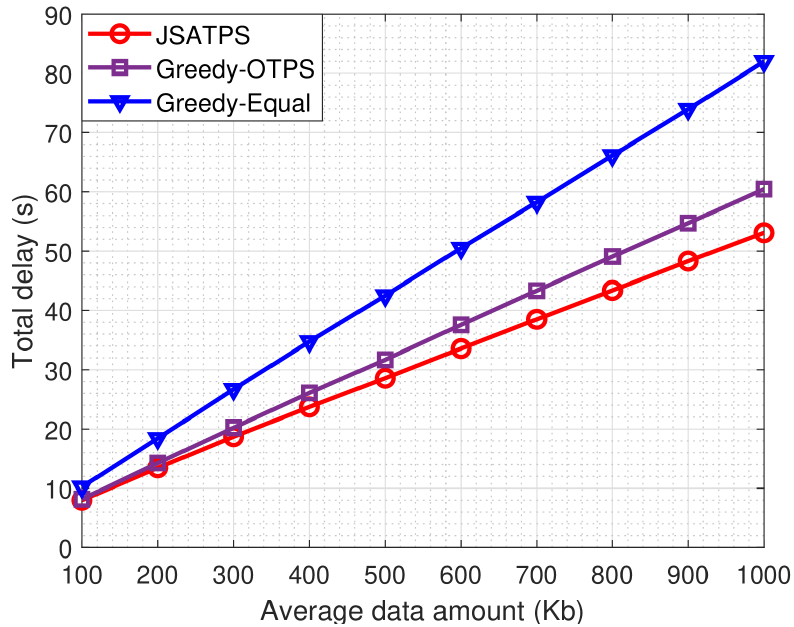}
\centerline{\footnotesize (b) $I^{total} = 1.2 \times 10^8$ bit.}
\end{minipage}
\captionsetup{font={footnotesize}}
\caption{Performance analysis of different average data amount.}\label{fig:data_amount}
\end{center}
\vspace{-0mm}
\end{figure*}

Fig. \ref{fig:data_amount} illustrates the performance analysis of different average data amount, in which one of the optimization objectives is fixed to give a more intuitive comparison. In Fig. \ref{fig:data_amount} (a), we investigate the achievable radar MI under different average data amount, when guaranteeing the task completion delay $t^{total}=$ 20 s. Note that the point $t^{total}=$ 20 s is not always on the Pareto frontier when the average data amount varies. Thus we only depict the feasible points on the Pareto frontier. It can be seen that the proposed JSATPS strategy can always achieve the largest radar MI for different cases. For average data amount of 500 Kb, the performance superiority of the proposed JSATPS strategy is 60 \% and 180 \% compared with the Greedy-OTPS strategy and the Greedy-Equal strategy. However, when the average data amount increases to as large as 600 Kb, it can be seen that the radar MI dramatically decreases to near zero. In this case, almost all the subframe resources are allocated for communication to guarantee the performance of the task completion delay. Also, note that the point $t^{total}=$ 20 s is not on the Pareto frontier of the Greedy-Equal strategy in this case. On the other hand, when the the average data amount is small enough, e.g. 200 Kb in the figure, local computing is almost sufficient to guarantee the task completion delay. Then, the subframe resources will be mainly allocated for sensing, and the Greedy-OTPS strategy can achieve similar performance compared with the JSATPS strategy in this case.
Similarly, in Fig. \ref{fig:data_amount} (b), we investigate the achievable task completion delay under different average data amount, when guaranteeing the radar MI $I^{total} = 1.2 \times 10^8$  bits. Results also show that the proposed JSATPS strategy is always superior to the comparison strategies. For average data amount of 500 Kb, the performance superiority of the proposed JSATPS strategy is 11 \% and 50 \% compared with the Greedy-OTPS strategy and the Greedy-Equal strategy. It can be seen that the effectiveness of the proposed strategy is validated in terms of both the radar MI and the task completion delay.

Fig. \ref{fig:orbit} illustrates the performance analysis of different orbit altitudes for the proposed JSATPS strategy. In Fig. \ref{fig:orbit} (a), the y axis at the left gives the achievable radar MI when guaranteeing the task completion delay $t^{total}=$ 20 s. It can be seen that the achievable radar MI decreases by 87 \% from $1.1 \times 10^8$ bits to $1.4 \times 10^7$ bits when the orbit altitude increases from 1,000 km to 10,000 km. Note that the point $t^{total}=$ 20 s is not on the Pareto frontier for orbit altitudes higher than 10,000 km. Similarly, the axis at the right gives the achievable task completion delay when guaranteeing the radar MI $I^{total} = 6 \times 10^7$  bits. The achievable task completion delay increases by 164 \% from 13 s to 34.3 s when the orbit altitude increases from 1,000 km to 36,000 km. Results indicate that the orbit altitude in the integrated satellite-terrestrial network has significantly impact on the network performance. In current satellite constellation projects, LEO satellites of low orbit altitude are generally more preferable to improve the network performance.

\begin{figure*} [tbp]
\begin{center}
\begin{minipage}[t]{0.33 \linewidth}
\centering
\includegraphics*[width=2.4in]{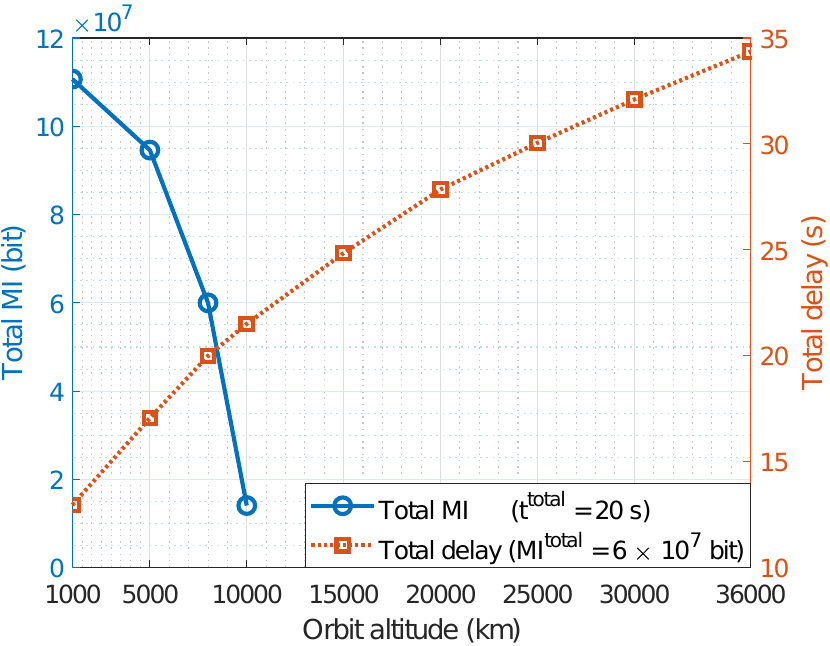}
\centerline{\footnotesize (a) Network performance analysis.}
\vspace{2mm}
\end{minipage}
\begin{minipage}[t]{0.33 \linewidth}
\centering
\includegraphics*[width=2.3in]{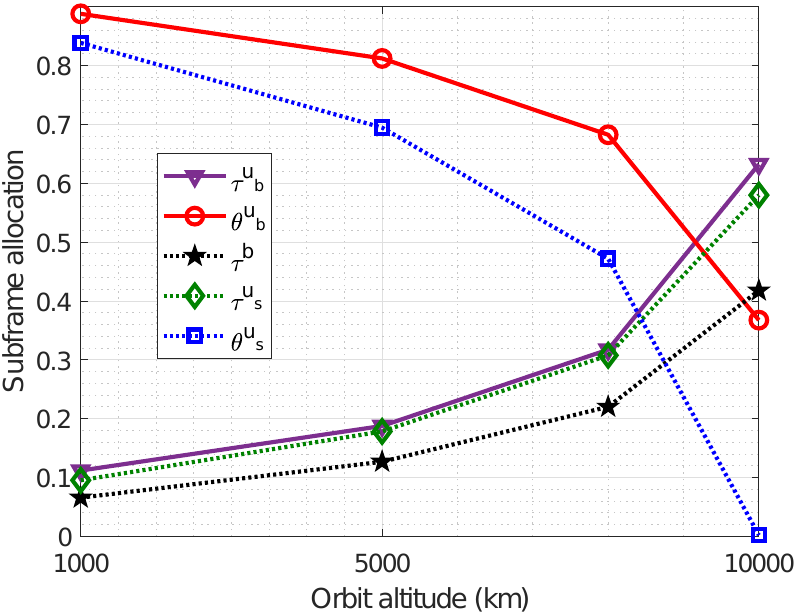}
\centerline{\footnotesize (b) Subframe allocation when $t^{total}= 20$ s.}
\end{minipage}
\begin{minipage}[t]{0.33 \linewidth}
\centering
\includegraphics*[width=2.3in]{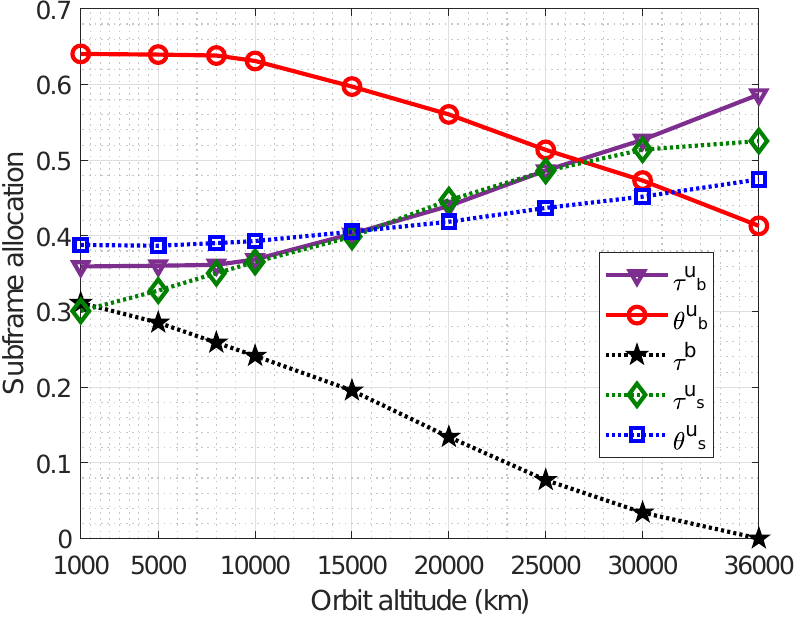}
\centerline{\footnotesize (c) Subframe allocation when $I^{total} = 6 \times 10^7$ bit.}
\end{minipage}
\vspace{-0mm}
\captionsetup{font={footnotesize}}
\caption{Performance analysis of different orbit altitudes.}\label{fig:orbit}
\end{center}
\vspace{-0mm}
\end{figure*}

In Fig. \ref{fig:orbit} (b), we investigate the subframe allocation strategy under different orbit altitudes, when guaranteeing the task completion delay $t^{total}=$ 20 s. For ease of depiction, the subframe allocation strategy of the BS is averaged by the $N$ BSs. First, we focus on the orbit altitude of 1,000 km, which is typical for LEO satellites. It can be seen that since the propagation delay of satellite-terrestrial links is relatively small, the performance of the task completion delay can be guaranteed with less communication resources. In this case, the communication subframe allocation strategy is given by $\tau^{u_b} =$ 0.11, $\tau^b =$ 0.07, $\tau^{u_s} =$ 0.10. Then, most of the subframe resources can be allocated for sensing, given by $\theta^{u_b,rad} =$ 0.89, $\theta^{u_s, rad} =$ 0.83. However, if the satellite is replaced by  MEO satellites on the orbit of 10,000 km, the propagation delay of satellite-terrestrial links will bring high burden to the task completion delay. As shown in the figure, much more subframe resources need to be allocated for communication to guarantee the performance of the task completion delay. The subframe allocated for communication will increases to $\tau^{u_b} =$ 0.63, $\tau^b =$ 0.42, $\tau^{u_s} =$ 0.58, while the subframe allocated for sensing decreases to $\theta^{u_b,rad} =$ 0.37, $\theta^{u_s, rad} =$ 0. It can be seen that the satellite subframe is completely allocated for communication in this case, since the communication burden is mainly brought by the satellite-terrestrial communication links of high altitude.

Then, in Fig. \ref{fig:orbit} (c), we investigate the subframe allocation strategy under different orbit altitudes, when guaranteeing the radar MI $I^{total} = 6 \times 10^7$ bits. Different from Fig. \ref{fig:orbit} (b), it can be seen that the communication subframe $\tau^b$ allocated to the BSs decreases as the orbit altitude increases. As discussed in Section \ref{sec:4}, whether the computation tasks of TUEs are offloaded to the cloud is determined by comparing the edge computing delay $t_{n,k}^{b,comp,edge}$ with the propagation delay $t_n^{b,trip}$. Since higher orbit altitude means higher propagation delay, less TUEs will be offloaded to the cloud. In order to maximize the network utility, it is more efficient to improve the edge computing performance of TUEs by increasing $\tau^{u_b}$ and decreasing $\theta^{u_b,rad}$. Then, the sensing subframe $\theta^{u_s,rad}$ allocated to SUEs needs to be increased to guarantee the performance of the radar MI. Also, note that the computation tasks of SUEs can only be executed locally or in the cloud. Due to the limited computing resources of users, it is still preferable to allocate more communication resources $\tau^{u_s}$ to SUEs when the orbit altitude increases.

\section{Conclusion}\label{sec:8}

In this paper, we have investigated the time-division based integrated sensing, communication, and computing in integrated satellite-terrestrial networks. Based on the proposed ISCC framework, we have formulated the joint subframe allocation and task partitioning problem to maximize the radar MI and minimize the task completion delay.
Based on the partial offloading model, the original optimization problem has been decomposed into the TUE task partitioning subproblem and the SUE task partitioning subproblem. The closed-form solutions have been obtained for both subproblems with theoretical derivations.
Then, we have developed the joint subframe allocation and task partitioning strategy to optimize the overall performance of the network. By leveraging the coordination of different network components, the sensing, communication, and computing can be efficiently integrated and managed in the satellite-terrestrial network.
The simulations results have demonstrated the effectiveness and superiority of the proposed strategies. Also, we have analyzed the trade-off between the radar MI and the task completion delay based on the Pareto frontier. It is expected that the results of the marginal effect, achievable performance, and resource utilization can help to provide constructive guideness for the network design.
More advanced physical layer techniques and artificial intelligence algorithms can be further explored to enhance the system performance in future works.

\section*{CRediT authorship contribution statement}
\textbf{Xiangming Zhu:} Methodology, Software, Investigation, Writing - Original Draft. \textbf{Hua Wang:} Conceptualization, Project Administration, Writing - Reviewing \& Editing. \textbf{Zhaohui Yang:} Supervision, Writing - Reviewing \& Editing. \textbf{Quoc-Viet Pham:} Writing - Reviewing \& Editing.

\section*{Acknowledgments}
This work was supported by the National Natural Science Foundation of China under Grant 62201513 and Grant 62271452.

\section*{Appendix}

\subsection*{A. Proof of Theorem \ref{theorem:task_partitioning_optimality}}\label{sec:App_1}

Let $\{ \alpha _{n,k}^{{u_b},cloud},\beta _{n,k}^{{u_b},cloud},\kappa _{n,k}^{{u_b},cloud}\} $ be the task partitioning strategy when the boundary condition is satisfied, and $t_{n,k}^{{u_b},total,cloud} = t_{n,k}^{{u_b},u,cloud} = t_{n,k}^{{u_b},b,cloud} = t_{n,k}^{{u_b},c,cloud}$ be the corresponding task completion delay. For task partitioning strategy $\{ \alpha _{n,k}^{{u_b}},\beta _{n,k}^{{u_b}},\kappa _{n,k}^{{u_b}}\}$, if the boundary condition is not satisfied, we have
\begin{align}\label{eq:Theorem1_boundary}
t_{n,k}^{{u_b},total} = \max \{ t_{n,k}^{{u_b},u},t_{n,k}^{{u_b},b},t_{n,k}^{{u_b},c}\}  > \min \{ t_{n,k}^{{u_b},u},t_{n,k}^{{u_b},b},t_{n,k}^{{u_b},c}\}.
\end{align}
If we can prove $t_{n,k}^{{u_b},total} > t_{n,k}^{{u_b},total,cloud}$, Theorem \ref{theorem:task_partitioning_optimality} is proved naturally.

We first assume $t_{n,k}^{{u_b},total} \le t_{n,k}^{{u_b},total,cloud}$, the following inequalities will hold in this case.
\begin{align}\label{eq:Theorem1_inequality}
&t_{n,k}^{{u_b},u} \le t_{n,k}^{{u_b},total} \le t_{n,k}^{{u_b},total,cloud} = t_{n,k}^{{u_b},u,cloud}\\
&t_{n,k}^{{u_b},b} \le t_{n,k}^{{u_b},total} \le t_{n,k}^{{u_b},total,cloud} = t_{n,k}^{{u_b},b,cloud} \nonumber \\
&t_{n,k}^{{u_b},c} \le t_{n,k}^{{u_b},total} \le t_{n,k}^{{u_b},total,cloud} = t_{n,k}^{{u_b},c,cloud}. \nonumber
\end{align}
Then, the following conclusions can be obtained.
\begin{itemize}
\item $\alpha _{n,k}^{{u_b}}$: The inequality $t_{n,k}^{{u_b},u} \le t_{n,k}^{{u_b},u,cloud}$ can be expressed as
    \begin{align}\label{eq:Theorem1_alpha}
    \frac{{\alpha _{n,k}^{{u_b}}D_{n,k}^{{u_b}}C_{n,k}^{{u_b}}}}{{f_{n,k}^{{u_b}}}} \le \frac{{\alpha _{n,k}^{{u_b},cloud}D_{n,k}^{{u_b}}C_{n,k}^{{u_b}}}}{{f_{n,k}^{{u_b}}}}.
    \end{align}
    We can conclude that $\alpha _{n,k}^{{u_b}} \le \alpha _{n,k}^{{u_b},cloud}$. Note that $\alpha _{n,k}^{{u_b}} = \alpha _{n,k}^{{u_b},cloud}$ only when $t_{n,k}^{{u_b},u} = t_{n,k}^{{u_b},u,cloud} = t_{n,k}^{{u_b},total,cloud}$.

\item $\beta _{n,k}^{{u_b}}$: The inequality $t_{n,k}^{{u_b},b} \le t_{n,k}^{{u_b},b,cloud}$ can be expressed as
    \begin{align}\label{eq:Theorem1_beta}
    &\frac{{(1 - \alpha _{n,k}^{{u_b}})D_{n,k}^{{u_b}}K_n^b}}{{R_{n,k}^{{u_b}}\tau _n^{{u_b}}}} + \frac{{\beta _{n,k}^{{u_b}}D_{n,k}^{{u_b}}C_{n,k}^{{u_b}}K_n^b}}{{{F^b}}} \\
    & \le \frac{{(1 - \alpha _{n,k}^{{u_b},cloud})D_{n,k}^{{u_b}}K_n^b}}{{R_{n,k}^{{u_b}}\tau _n^{{u_b}}}} + \frac{{\beta _{n,k}^{{u_b},cloud}D_{n,k}^{{u_b}}C_{n,k}^{{u_b}}K_n^b}}{{{F^b}}}. \nonumber
    \end{align}
    Since we have proved $\alpha _{n,k}^{{u_b}} \le \alpha _{n,k}^{{u_b},cloud}$, we can conclude that $\beta _{n,k}^{{u_b}} \le \beta _{n,k}^{{u_b},cloud}$. Note that $\beta _{n,k}^{{u_b}} = \beta _{n,k}^{{u_b},cloud}$ only when $\alpha _{n,k}^{{u_b}} = \alpha _{n,k}^{{u_b},cloud},t_{n,k}^{{u_b},b} = t_{n,k}^{{u_b},b,cloud} = t_{n,k}^{{u_b},total,cloud}$.

\item $\kappa _{n,k}^{{u_b}}$: We have proved that $\alpha _{n,k}^{{u_b}} \le \alpha _{n,k}^{{u_b},cloud},\beta _{n,k}^{{u_b}} \le \beta _{n,k}^{{u_b},cloud}$. Based on the constraint $\alpha _{n,k}^{{u_b}} + \beta _{n,k}^{{u_b}} + \kappa _{n,k}^{{u_b}} = 1$, we have
    \begin{align}\label{eq:Theorem1_kappa}
    \kappa _{n,k}^{{u_b}} & = 1 - \alpha _{n,k}^{{u_b}} - \beta _{n,k}^{{u_b}} \\
    & \ge 1 - \alpha _{n,k}^{{u_b},cloud} - \beta _{n,k}^{{u_b},cloud} = \kappa _{n,k}^{{u_b},cloud} > 0. \nonumber
    \end{align}
    We can conclude that task $Q_{n,k}^{{u_b}}$ is offloaded to the cloud.

\item $t_{n,k}^{{u_b},c}$: We have proved that $\alpha _{n,k}^{{u_b}} \le \alpha _{n,k}^{{u_b},cloud},\beta _{n,k}^{{u_b}} \le \beta _{n,k}^{{u_b},cloud},$ $\kappa _{n,k}^{{u_b}} > 0$. Then, we have
    \begin{align}\label{eq:Theorem1_t}
    t_{n,k}^{{u_b},c} &= \frac{{(1 - \alpha _{n,k}^{{u_b}})D_{n,k}^{{u_b}}K_n^b}}{{R_{n,k}^{{u_b}}\tau _n^{{u_b}}}} \\
    &\;\;\; + \frac{{(1 - \alpha _{n,k}^{{u_b}} - \beta _{n,k}^{{u_b}})D_{n,k}^{{u_b}}{K^{{u_b},c}}}}{{R_n^b{\tau ^b}}} + t_n^{b,trip} \nonumber\\
    &\ge \frac{{(1 - \alpha _{n,k}^{{u_b},cloud})D_{n,k}^{{u_b}}K_n^b}}{{R_{n,k}^{{u_b}}\tau _n^{{u_b}}}} \nonumber \\
    &\;\;\; + \frac{{(1 - \alpha _{n,k}^{{u_b},cloud} - \beta _{n,k}^{{u_b},cloud})D_{n,k}^{{u_b}}{K^{{u_b},c}}}}{{R_n^b{\tau ^b}}} + t_n^{b,trip} \nonumber \\
    & = t_{n,k}^{{u_b},c,cloud} = t_{n,k}^{{u_b},total,cloud}. \nonumber
    \end{align}
    Note that $t_{n,k}^{{u_b},c} = t_{n,k}^{{u_b},c,cloud} = t_{n,k}^{{u_b},total,cloud}$ only when $\alpha _{n,k}^{{u_b}} = \alpha _{n,k}^{{u_b},cloud},\beta _{n,k}^{{u_b}} = \beta _{n,k}^{{u_b},cloud}$.

\end{itemize}

In (\ref{eq:Theorem1_inequality}), the inequality $t_{n,k}^{{u_b},c} \le t_{n,k}^{{u_b},total,cloud}$ is derived based on the assumption $t_{n,k}^{{u_b},total} \le t_{n,k}^{{u_b},total,cloud}$. Also, we have proved $t_{n,k}^{{u_b},c} \ge t_{n,k}^{{u_b},total,cloud}$ in (\ref{eq:Theorem1_t}). Then, we can conclude that $t_{n,k}^{{u_b},c} = t_{n,k}^{{u_b},total,cloud}$, based on which we have $\alpha _{n,k}^{{u_b}} = \alpha _{n,k}^{{u_b},cloud},\beta _{n,k}^{{u_b}} = \beta _{n,k}^{{u_b},cloud}$ as analyzed in (\ref{eq:Theorem1_t}). Then, we can conclude $t_{n,k}^{{u_b},u} = t_{n,k}^{{u_b},total,cloud},t_{n,k}^{{u_b},b} = t_{n,k}^{{u_b},total,cloud}$ as analyzed in (\ref{eq:Theorem1_alpha}) and (\ref{eq:Theorem1_beta}). The conclusion $t_{n,k}^{{u_b},u} = t_{n,k}^{{u_b},b} = t_{n,k}^{{u_b},c} = t_{n,k}^{{u_b},total,cloud}$ is finally obtained. However, it conflicts with the assumption in (\ref{eq:Theorem1_boundary}). Thus, the inequality $t_{n,k}^{{u_b},total} > t_{n,k}^{{u_b},total,cloud}$ will always hold if the boundary condition $t_{n,k}^{{u_b},u} = t_{n,k}^{{u_b},b} = t_{n,k}^{{u_b},c}$ is not satisfied. Theorem \ref{theorem:task_partitioning_optimality} is proved.

\end{document}